\documentclass[fleqn,10pt]{wlscirep}
\usepackage[utf8]{inputenc}
\usepackage[T1]{fontenc}
\usepackage{subcaption}
\usepackage[utf8]{inputenc}
\usepackage{amsthm}
\usepackage{hyperref}
\usepackage{algorithm}
\usepackage{algpseudocode}
\usepackage[utf8]{inputenc} 
\usepackage{amsmath} 
\usepackage{comment}
\usepackage{listings}
\usepackage{xcolor}
\usepackage[left]{lineno}

\usepackage{amsthm}

\newtheorem{theorem}{Theorem}
\newtheorem{construction}{Construction}
\theoremstyle{definition}

\newtheorem{definition}{Definition}

\title{Error characterization and error correction approaches in combinatorial DNA-based storage}

\author[1,2,*,+]{Inbal Preuss}
\author[2,*,+]{Omer Sabary}
\author[3]{Ryan Gabrys}
\author[1,2]{Zohar Yakhini}
\author[2]{Eitan Yaakobi}
\author[1]{Leon Anavy}
\affil[1]{Reichman University, Computer Science, Herzliya, Israel}
\affil[2]{Technion - Israel Institute of Technology, Computer Science, Haifa, Israel}
\affil[3]{Department of Electrical and Computer Engineering, University of California San Diego, La Jolla, CA, USA}

\affil[*]{Emails: inbalpreuss@gmail.com, omer.orange@gmail.com}
\affil[+]{these authors contributed equally to this work}

\begin{abstract}
Data storage in DNA has recently emerged as a promising archival solution, offering space-efficient and long-lasting digital storage. DNA’s ultra-high density and durability make it an attractive medium for long-term data storage. Among the various approaches, combinatorial DNA encoding further enhances this potential by increasing the logical density through the use of combinations of DNA shortmers, where each sequence position is represented by a set of predefined short DNA fragments. This approach allows for the encoding of a larger data volume using fewer synthesis cycles. However, this method introduces unique challenges, particularly in terms of synthesis and sequencing errors.

In this study, we focus on the characterization of errors in combinatorial DNA-based storage systems. Our analysis revealed that asymmetric combinatorial erasure errors, defined as the omission of a single shortmer from the set defining the combinatorial letter, are a prevalent error type in combinatorial DNA-based storage, particularly in large-scale systems where read coverage is limited. We analyzed two previously published datasets, observing a high frequency of erasure errors, where missing sequences obstruct the reconstruction of specific combinatorial letters.

To better understand these observations, we conducted a large-scale experimental proof-of-concept of a combinatorial DNA-based storage system and evaluated the error characteristics of the system. Our analysis confirmed that erasure errors become increasingly prominent with reduced sequencing depth. We demonstrated that below 50 reads per sequence, the frequency of erasure errors sharply increased. We developed an asymmetric error-correcting code specifically designed to address these errors. The code utilizes tensor-product (TP) codes to integrate standard erasure and substitution-correcting codes (such as Reed-Solomon (RS) codes) with Varshamov-Tenengolts (VT) codes, which are asymmetric error-correcting codes.

We validated the performance of our new error correction code both in simulations and in a second large-scale experiment. The experimental comparison was designed to directly compare our suggested code with the more straightforward 2D Reed-Solomon (2D RS) scheme. Our method consistently outperformed the 2D RS scheme, particularly in scenarios dominated by erasure errors. Notably, in the second large-scale experiment, our method demonstrated superior decoding accuracy even under low coverage conditions, where the traditional 2D RS approach struggled to decode the data.

Our findings demonstrate the importance of tailored error correction schemes in DNA-based data storage. By directly addressing the asymmetric nature of errors in combinatorial DNA, our method provides improved decoding accuracy under a wide range of conditions. The integration of such tailored error correction methods with existing DNA-based data storage technologies has the potential to deliver more reliable and scalable DNA-based data applications.

\end{abstract}
\begin{document}
\flushbottom
\maketitle
%
%
\thispagestyle{empty}

\section*{Introduction}
The global data landscape is expanding rapidly, with estimates placing it at 180 zettabytes by the end of 2025 and projections suggesting it will exceed 400 zettabytes by 2028~\cite{taylor2023amount}. Current storage devices, such as magnetic tapes, hard disks, and solid-state drives, are struggling to meet this demand~\cite{Forbes, RGR18}, mostly because of limitations in physical space and energy consumption. As an alternative solution, DNA-based data storage systems has been suggested, offering high density, long-span durability, and low maintenance energy consumption.  Since 2012, various studies presented successful DNA-based storage solutions, both in academia and industry~
\cite{church2012next, grass2015robust, goldman2013towards,yazdi2017portable, erlich2017dna, organick2018random, anavy2019data, bar2024zettabyte, sabary2024survey,bar2025scalable}. In addition to these works, several recent studies have proposed new encoding approaches for DNA-based data storage. Wu \textit{et al.}~\cite{wu2025stable} introduced the Repeating Substring Tree Encoding (RSTE) method, which identifies repeated substrings in the source data and encodes them as DNA motifs to improve the encoding rate while ensuring that the resulting sequences satisfy standard biochemical constraints such as run-length limits, GC balance, and end constraints. Other examples include a parity encoding and local mean iteration (PELMI) scheme for robust DNA image storage~\cite{cao2024PELMI} and an invertible neural network–based DNA image storage method (INNSE) that applies up–down sampling to reduce the number of required DNA sequences and uses an internal base-level self-correction mechanism without additional redundancy~\cite{liu2025family}.

In general, the in-vitro DNA-based storage pipeline works as follows. First, the user's information is encoded into DNA sequences (sequences over the alphabet of ${A, C, G, T}$) based on a predefined coding scheme. These sequences are termed \emph{encoded sequences}. Next, molecules representing the encoded sequences are generated using a biochemical process termed \emph{DNA synthesis}, which produces single-stranded DNA oligonucleotides (or in short, \emph{oligos}). These synthesized oligos, which now store the user's information, are placed together, typically unordered, in a small storage container, usually a vial. To retrieve the data, a small sample of the oligos is taken, and using \emph{DNA sequencing}, the nucleotide base sequences are read back and decoded. Since both synthesis and sequencing introduce errors, it is necessary to encode the information with an \emph{error-correcting code (ECC)}~\cite{chandak2019improved, anavy2019data, bar2025scalable, organick2018random, yazdi2017portable, erlich2017dna}. Thus,  by including redundancy within the encoded sequences, it is possible to create a system that maintains data reliability even in the presence of errors.

Biological and biochemical constraints are fundamental to the design of reliable DNA-based data storage systems. Synthesized oligos must satisfy requirements such as maintaining appropriate GC content, avoiding long homopolymer runs, minimizing secondary structure formation, and preventing unintended cross-hybridization or self-priming during PCR and sequencing. These considerations have motivated a range of constrained coding techniques, including mutually uncorrelated (MU) codes and k-weakly mutually uncorrelated (k-WMU) codes, which restrict long prefix–suffix overlaps and are commonly used to design unique sequences that improve access efficiency in DNA-based storage systems. Recent work by Liu \textit{et al.}~\cite{liu2025family} presented systematic constructions of MU and k-WMU code families that explicitly account for multiple biological and biochemical constraints during sequence design. While our work focuses on combinatorial DNA encoding, these biological and biochemical constraints must also be taken into consideration in the design and analysis of combinatorial DNA-based storage systems.

A major bottleneck that still limits the adoption of DNA-based data storage solutions is the cost and latency of DNA synthesis~\cite{antkowiak2020low}. On the other hand, most of the synthesis methods have inherent redundancy, which is reflected by the fact they generate multiple copies per oligo, usually in the order of thousands to hundred thousands~\cite{leproust2010synthesis}. To leverage this bottleneck, the idea of \emph{Composite DNA} has been suggested in 2019 by Anavy \textit{et al.}~\cite{anavy2019data} and by Choi et al~\cite{choi2019high}. These methods use the inherent redundancy of the synthesis to introduce more alphabet letters beyond the standard $\{A, C, G, T\}$. Therefore, the potential information density increases, meaning that more information can be encoded using the same number of synthesis cycles. Informally speaking, a composite DNA letter is a mixture consisting of more than one nucleotide in a predetermined ratio. This predetermined mixture is represented as randomly synthesizing different oligos on different molecules in the same position. Thus, to identify the composite letter, it is required to sequence a fraction of these multiple copies and to estimate the original composition of the different nucleotides. Although this method increases capacity, encoding using composite DNA introduces new errors and thus requires specified error correction codes. Error-correcting codes for composite DNA were suggested by Walter et. al~\cite{walter2024coding}, while results related to the information capacity of this channel were given by Cohen \textit{et al.}~\cite{cohen2024optimizing} and Kobovich \textit{et al.}~\cite{kobovich2023m, kobovich2025deepdive}.  

Another method to increase logical density and reduce overall costs in DNA-based storage systems is to use the combinatorial DNA synthesis method, which  have recently been suggested \cite{preuss2024efficient, roquet2021dna, yan2023scaling}.
Molecular mechanisms for generating combinatorial sets of DNA sequences include combinatorial DNA synthesis and combinatorial DNA assembly. Nonetheless, they all share a common encoding approach of using combinations of DNA sequences as the data-holding unit in the system. 
These common characteristics was previously referred to as \textit{combinatorial DNA encoding} and formally defined as follows \cite{preuss2024efficient}. Let $\Omega$ be a set of $N$ DNA k-mers that are used as building blocks and selected to minimize mix-up errors (\textit{e.g.,} by ensuring a minimal distance $d$ between each pair).
A combinatorial alphabet $\Sigma$ is defined  such that each letter $\sigma \in \Sigma$ represents a subset of $K$ k-mers from $\Omega$, forming an alphabet of size $|\Sigma| \leq \binom{N}{K}$.
A binary message of length $B$ bits is encoded by generating a sequence of $M$ combinatorial letters $\sigma^1\cdots\sigma^M$, where $M = \frac{B}{\lfloor \log_{2}(|\Sigma|) \rfloor}$. Fig.~\ref{fig:view of a combinatorial alphabet} illustrates this combinatorial encoding approach presenting an example combinatorial alphabet with $N=8$ and $K=4$ and the encoding of a 42 bits message using a 7 letter combinatorial sequence.
After encoding the input binary message into a set of combinatorial DNA sequences, various molecular techniques can be used to generate a pool of DNA molecules. Each actual oligo or molecule holds a single combination of the DNA k-mers and collectively, the pool represents the original set of combinatorial sequences. 
After sequencing a sample of the DNA molecules, reconstruction of the combinatorial sequences involves three key steps. First, the reads are split into groups where each group represents a single combinatorial sequence. This step is typically done using a barcode sequence located at a predefined position of each DNA molecule. Next, each combinatorial sequence is reconstructed independently from the grouped reads, one position at a time, by detecting $K$ unique k-mers in the relevant position and identifying the corresponding combinatorial letter. Lastly, the original message is decoded by applying error-correcting codes (Refer to Fig. 2 in \cite{preuss2024efficient} for details on the combinatorial encoding pipeline).

\begin{figure}
    \centering
    \includegraphics[width=0.9\linewidth]{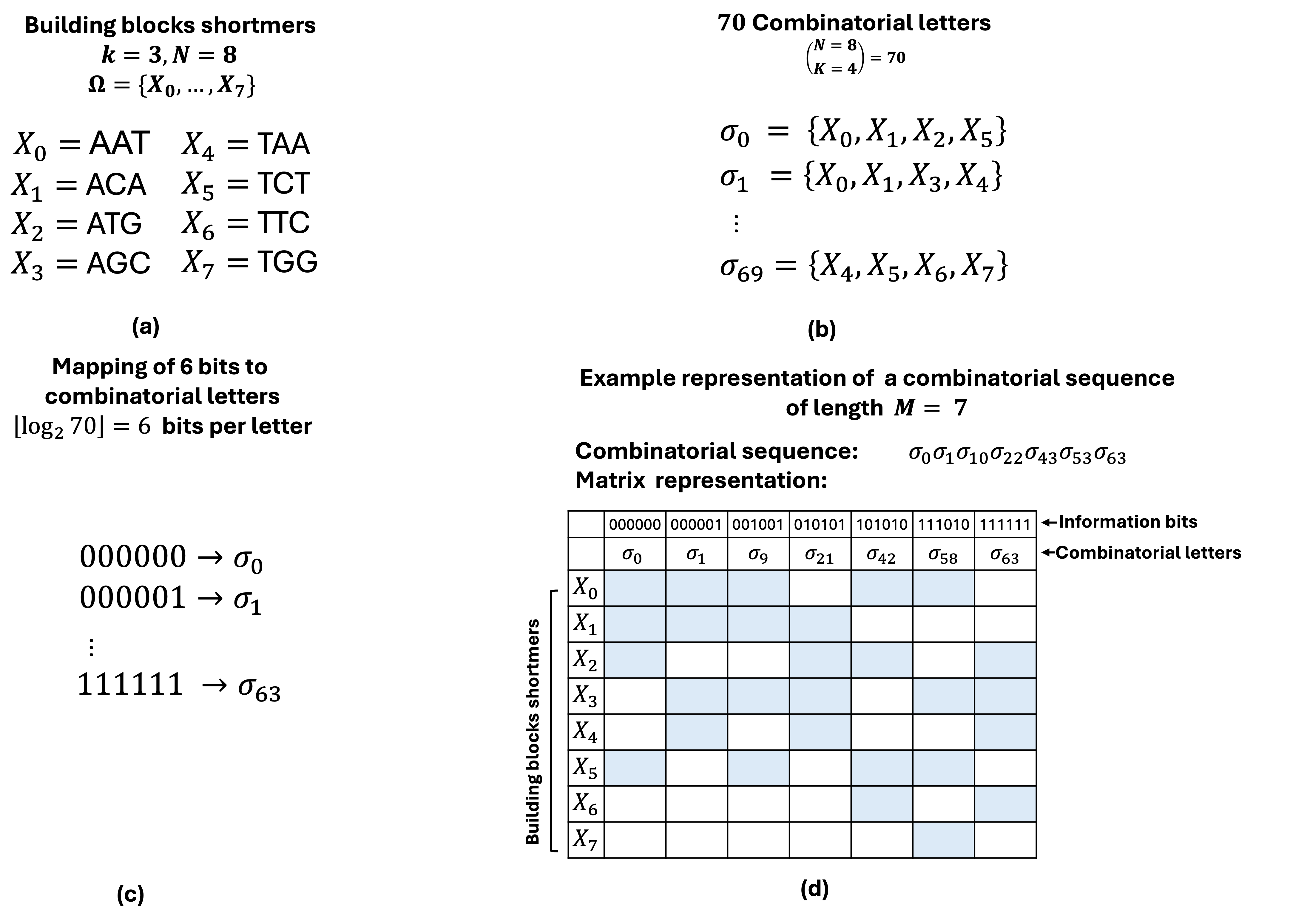}
    \caption{Schematic overview of a combinatorial alphabet and encoding. (a) The plot describes  $\Omega = \{ X_0, \ldots, X_7\}$, the set of building block shortmers of length $k=3$. (b) The building block shortmers can create a set of $70$ combinatorial letters, each is defined a selection of $K=4$ shortmers. (c) Since $\lfloor \log_2(70) \rfloor =6$, all the $2^6=64$ possible vectors of 6 bits are mapped into $64$ combintorial letters. (d) An example of a combinatorial sequence of length $M=7$, represented as a matrix. Each column corresponds to a combinatorial letter, and each row represents a shortmer building block. The cells corresponding to the selected shortmers in each column are highlighted in light blue.}
    \label{fig:view of a combinatorial alphabet}
\end{figure}

Standard DNA-based storage systems use error correction to address common synthesis and sequencing errors, such as single-base substitutions, insertions, and deletions. These symbol-level errors are addressed in most studies by applying error correction codes and constrained coding. Another frequent issue in DNA-based data storage systems is sequence dropout, where certain sequences are altogether missing from the NGS output, often due to biases in molecular biology processes and the effect of sampling in the readout process. To address this, sequence-level error correction codes are used allowing the recovery of complete missing or corrupt sequences from the message. Combining symbol-level (inner) and sequence-level (outer) error correction codes forms a 2-dimensional (2D) error correction scheme. A common method for addressing both symbol and sequence-level errors is leveraging the inherent multiplicity in DNA synthesis and sequencing. Increasing the sampling rate improves sequence coverage, reduces dropouts, and corrects symbol-level errors through consensus-based methods~\cite{bar2024cover,abraham2024covering}. However, when considering large scale systems, the ability to increase sequence coverage may be limited and also caries significant costs. Preuss \textit{et al.}~\cite{preuss2024sequencing} applied a variant of the coupon collector distribution to analyze and determine the required sequencing coverage in combinatorial DNA-based data storage systems. In contrast to standard DNA-based systems, the reconstruction of a combinatorial letter depends on 
observing at least one copy of each k-mer of the 8 composition, making sequence coverage a critical parameter. Even after observing $K$ unique k-mers, we are not guaranteed to recover the correct combinatorial letter due to possible errors. A proposed solution is to observe each k-mer multiple times before inferring the combinatorial letter from the $K$ most frequent k-mers. This required multiplicity ($t$) is adjustable and its effect on the overall required sequencing depth is also analyzed in literature~\cite{bar2024cover}.

Optimal error correction schemes rely on a proper characterization of the information channel and the expected error types. To properly characterize the combinatorial DNA channel, we analyzed two previously published datasets focusing on the observed error patterns, particularly when down-sampling the data. Our analysis revealed that erasure errors, where combinatorial letters with missing k-mers are observed, were a dominant error type in both datasets, which represent different synthesis protocols and sequencing technologies. While these “asymmetric combinatorial erasures” can arise as a direct consequence of base-level insertions, deletions, or substitutions that render certain shortmers unidentifiable, our analysis indicates that they also reflect higher-level misrepresentation events specific to the combinatorial structure, such as the absence of a shortmer in a given cycle. This observation highlighted the critical impact of read depth on decoding accuracy in combinatorial DNA data storage.  Given the nature of some combinatorial DNA experimental systems, where high variations in sequence length and structure are observed, we developed a BLAST-based approach for sequence analysis. This method enabled accurate identification of k-mers within each read, even in cases where synthesis and sequencing errors disrupted the expected sequence structure.
To further investigate these observations we developed, together with Helixworks Technologies Ltd, a large-scale experimental proof-of-concept of a combinatorial DNA-based data storage system. The experiment, including 640 combinatorial sequences synthesized using 8 cycles, represents the first implementation of a large-scale combinatorial DNA-based storage system. We used the data from this experiment to evaluate the error characteristics of combinatorial DNA. Our analysis of the sequencing data confirmed a high prevalence of erasure errors, particularly when the read depth was reduced. We demonstrated that when the number of sequencing reads per sequence is low, the frequency of erasure errors sharply increased, significantly impacting decoding accuracy.
In response to these findings, we developed an asymmetric error-correcting code, specifically designed to address asymmetric errors in combinatorial DNA-based data storage \cite{sabary2024error}. This k-mer level error correction is introduced on top of every base-level  error correction that is being used to specifically mitigate the remaining asymmetric combinatorial errors. We compared this error correction scheme to existing error correction approaches in combinatorial DNA-based storage using simulations. Since the combinatorial DNA channel differs fundamentally from the nucleotide-level channels considered in most prior DNA storage works, we restrict quantitative comparisons to coding schemes operating under comparable combinatorial alphabets and error models.
Under these conditions,  our method consistently outperformed 2-dimensional Reed-Solomon approaches (2D RS), particularly in scenarios dominated by erasure errors, reflecting the real-world conditions observed in the analyzed datasets.
Finally, we performed a second large scale experiment, which was a direct experimental comparison using 320 sequences to evaluate the performance of our new error correction method against 320 sequences of the 2D RS approach. Our method demonstrated superior performance in recovering sequences with erasure errors, particularly under low coverage conditions, where 2D RS struggled.  These results validate the effectiveness of our approach, establishing a clear advantage over traditional methods in real-world scenarios. Our study not only provides a detailed understanding of the error characteristics in combinatorial DNA-based data storage but also offers a practical, scalable error correction solution for improving data integrity in the context of this emerging technology.

\section*{Results} \label{sec:results}

\subsection*{Error Characterization of Previously Published Datasets}
\phantomsection
\label{sec:data-analysis-error-characterization}

To characterize the channel properties of combinatorial DNA based storage systems, we analyzed two previously published datasets ~\cite{yan2023scaling, preuss2024efficient}. While both experiments demonstrated the encoding of data using combinatorial DNA, they differ in their sequencing technologies and synthesis protocols; Yan ~\textit{et al} use Bridge Oligo Assembly and Preuss~\textit{et al} use Gibson assembly. Table~\ref{tab:experiment-comparison} includes the design parameters for both experiments.

\begin{table}[ht]
\centering
\begin{tabular}{|l|l|l|}
\hline
\textbf{Parameter} & \textbf{Yan \textit{et al.} \cite{yan2023scaling}} & \textbf{Preuss \textit{et al.} \cite{preuss2024efficient}} \\ \hline
Message length (bits) & 84 & 96 \\ \hline
Sequence length/number of cycles (M) & 1 & 4 \\ \hline
Number of sequences & 8 (repeats) & 2 \\ \hline
Logical density (bits/synthesis-cycle) & 84 & 12 \\ \hline
Number of unique k-mers (N) & 96 & 16 \\ \hline
Number of unique k-mers in each letter (K) & 32 & 5 \\ \hline
Sequence length in bases & 74 bp & 220 bp \\ \hline
k-mer length & 25 bp & 20 bp \\ \hline
Synthesis protocol & Bridge Oligo Assembly & Gibson Assembly \ \ \ \ \  \\ \hline
Sequencing method & Nanopore & Illumina (MiSeq) \\ \hline
\end{tabular}
\caption{Design parameters for the datasets analyzed in this section.}
\label{tab:experiment-comparison}
\end{table}

We start by analyzing the data from Yan \textit{et al.}~\cite{yan2023scaling}. The experiment included a single combinatorial synthesis cycle (\(M=1\)) in which two DNA shortmers (\textit{i.e.} a barcode and a single payload) are assembled using Bridge Oligonucleotide Assembly.   The payload shortmer represents a single combinatorial letter with \(N=96\) and \(K=32\) encoding an 84-bit "HelloWorld" message. This experiment was repeated to generate 8 different sequences where each sequence contains a unique barcode sequence (BC), a universal sequence (U1), and a combinatorial payload sequence (P1). {The universal sequence (U1) is a short, fixed sequence that is identical across all sequences in this cycle and is used in the assembly process. Typically, each cycle has its own distinct universal sequence, which we denote by Ui, where $i$ corresponds to the cycle index.}
The sequence design in depicted in Fig.\ref{fig:helloWorld-a}. The resulting molecules were sequenced using Oxford Nanopore MinION technology, generating 102,222 reads to be analyzed.

First we examined read length distribution of the analyzed data (See Fig.~\ref{fig:helloWorld-b}) and observed an average read length of 174 nt with a standard deviation of 62 nt and a median length of 164 nt. We filtered the reads to retain only those longer than 50 nucleotides. Since many of the reads were much longer than the expected 74 nt, locating the universal sequence required a dedicated algorithm.   

For that, we employed Basic Local Alignment Search Tool (BLAST) \cite{altschul1990basic} to locate the U1 sequence position in each read. We then used this position for matching the BC and payload sequences adjacent to it. Details on this process are further elaborated in Methods~\hyperref[Methods:Data-Analysis-HelloWorld-1-cycles]{BLAST-based analysis of the data from Yan \textit{et al.}\cite{yan2023scaling}}.

Analysis of the start positions of U1, depicted in Fig.~\ref{fig:helloWorld-c}, revealed that the U1 start positions varied significantly across the reads (average = 174, median = 164, sd = 62). This variability further justified the need to use the a dedicated algorithm approach to identify the universal start position using BLAST-based algorithm. The reads were then split according to the identified barcode, distinguishing the 8 repeats, resulting in an average of $\sim$1,000 reads per sequence. After keeping only reads in which both the barcode sequence and a valid payload sequence were identified, an effective dataset of only $\sim$150 reads per sequence was generated. See sequence counts in Supplementary Table 1. As illustrated in Fig.~\ref{fig:helloWorld-d} these conditions led to a failure in recovering all 32 shortmers, highlighting significant deficiencies in read quantity and quality while also underscoring the challenges in data recovery despite complex decoding algorithms.

\begin{figure*}[htbp]
  \centering
  \begin{subfigure}[b]{0.48\linewidth}
    \centering
    \includegraphics[width=0.48\linewidth]{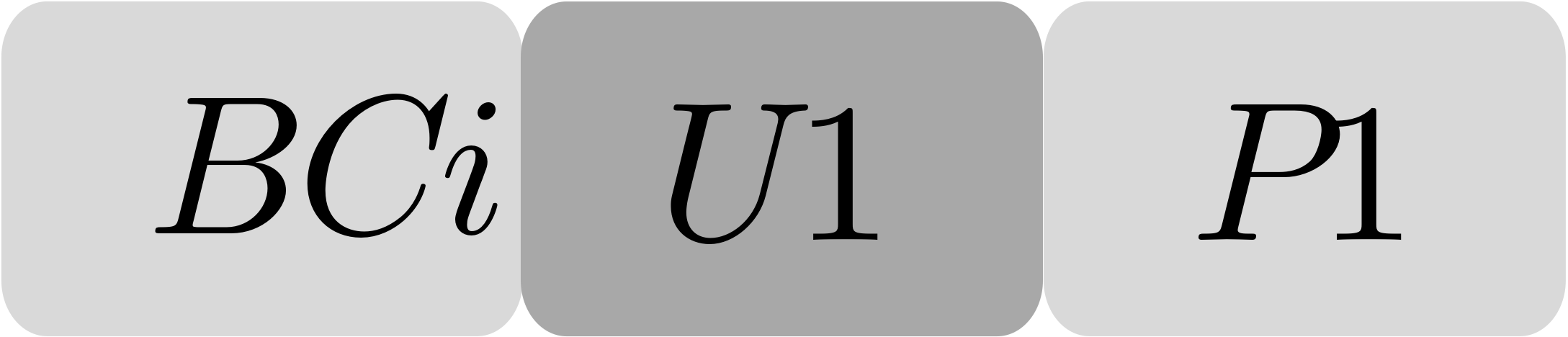}
    \caption{}
    \label{fig:helloWorld-a}
  \end{subfigure}
  \hfill
  \begin{subfigure}[b]{0.48\linewidth}
    \centering
    \includegraphics[width=\linewidth]{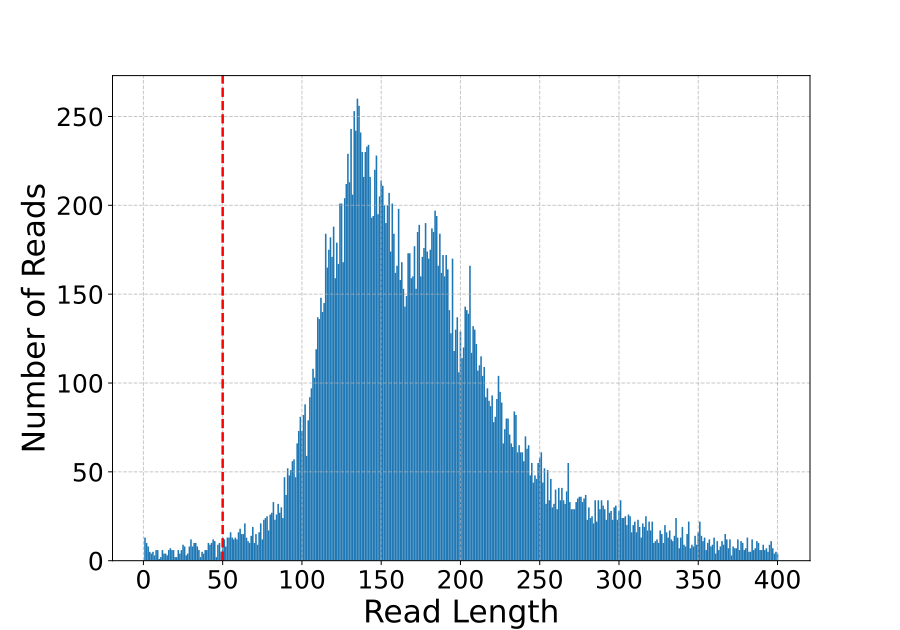}
    \caption{}
    \label{fig:helloWorld-b}
  \end{subfigure}
  
  \vspace{0.5em}
  
  \begin{subfigure}[b]{0.48\linewidth}
    \centering
    \includegraphics[width=\linewidth]{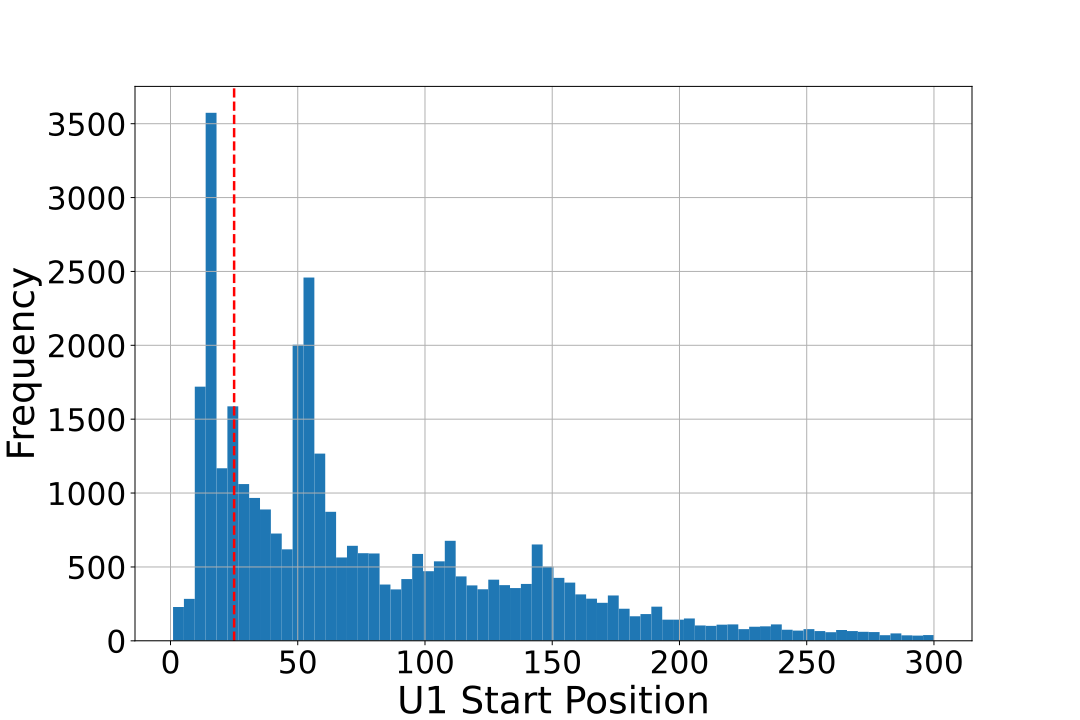}
    \caption{}
    \label{fig:helloWorld-c}
  \end{subfigure}
  \hfill
  \begin{subfigure}[b]{0.48\linewidth}
    \centering
    \includegraphics[width=\linewidth]{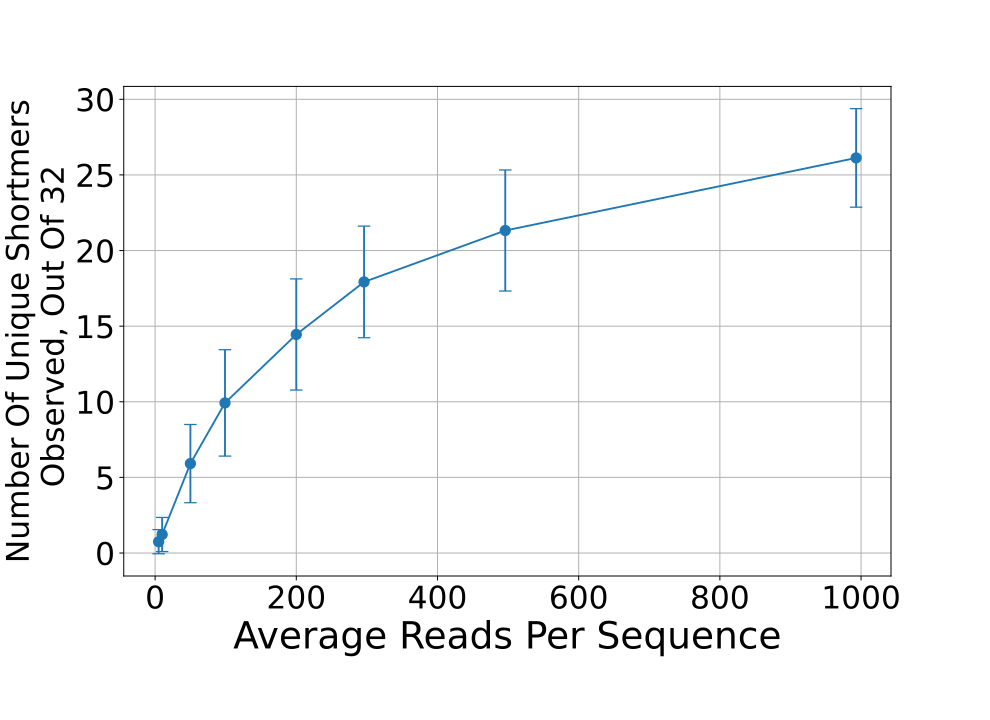}
    \caption{}
    \label{fig:helloWorld-d}
  \end{subfigure}
  \caption{Analysis of experimental data from Yan \textit{et al.}~\cite{yan2023scaling}. (a) Sequence design: each construct consists of a unique barcode (BCi), a universal sequence (U1), and a payload segment (P1). The universal sequence is fixed across all constructs in the same cycle and facilitates the assembly process.
 (b) Read length distribution, the dotted red line represents the filtering threshold used for analysis. (c) Distribution of the start position of U1 in the analyzed sequences. The dotted red line represents the designed start position. (d)  Evidence for erasure errors. X-axis is the average reads per sequence, and y-axis is the number of  unique shortmers observed out of 32, with $N=96$, $K=32$, $M=1$, 8 different sequences, the error bars are calculated over 10 repeats of randomly down-sampling the actual data.}
  \label{fig:HellowWorld}
\end{figure*}

Next, we analyzed data taken from the study by Preuss \textit{et al.}~\cite{preuss2024efficient}, which used a combinatorial Gibson assembly protocol and NGS sequencing using Illumina Miseq. This experiment included two combinatorial sequences, encoding 96 bits, which code the message “DNA Storage!”. Each sequence comprised of \(M=4\) cycles, using a combinatorial alphabet with \(N=16\), and \(K=5\) unique shortmers. This resulted in DNA molecules of length 220 nucleotides, each containing 6 universal sequences. The sequencing output included 2,634,683 total reads, of which 2,151,661 were of the expected length of 220 nucleotides (average = 217, median = 220, sd = 10, see Fig.~\ref{fig:dnastorage-a}).

Our analysis focused only on the sequences of length 220. For the first sequence, which was barcoded by BC1, there were 1,365,295 sequences, and for the second sequence (barcoded by BC2), we had 768,755 sequences, as detailed in Supplementary Table 1. The high quality of the sequencing results in this experiment made the retrieval of the barcode, universal sequences, and the 5 payload sequences in each cycle relatively easy. As a result, the original message was successfully recovered even after down-sampling the sequencing output to include an average of 100 reads per sequence (See Fig.~\ref{fig:dnastorage-b}). Further reducing the sequencing depth resulted in a failure to recover the original message. This was caused by a failure to detect all 5 shortmers in each cycle.

\begin{figure*}[htbp]
  \centering
  \begin{subfigure}[b]{0.48\linewidth}
    \centering
    \includegraphics[width=\linewidth]{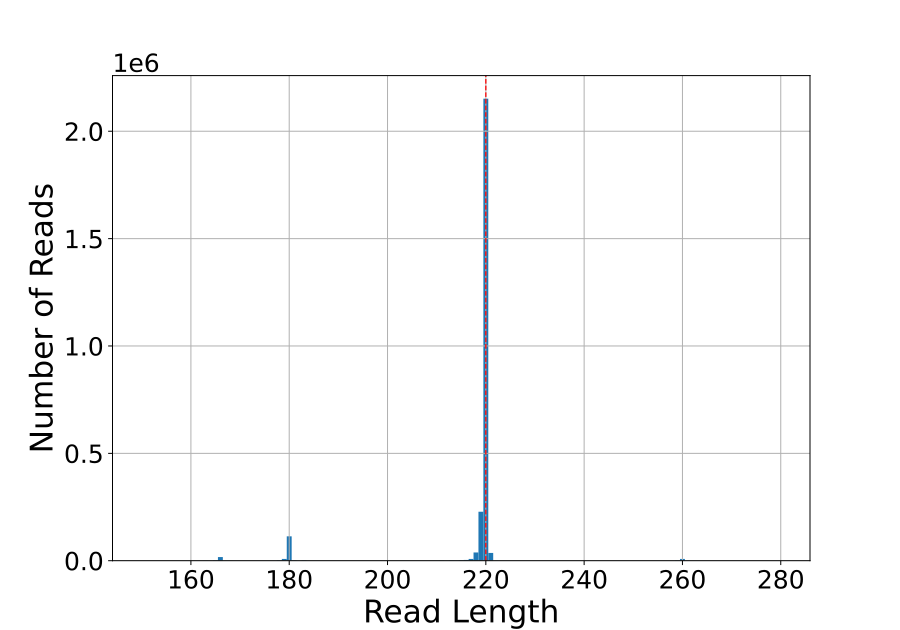}
    \caption{}
    \label{fig:dnastorage-a}
  \end{subfigure}
  \begin{subfigure}[b]{0.48\linewidth}
    \centering
    \includegraphics[width=\linewidth]{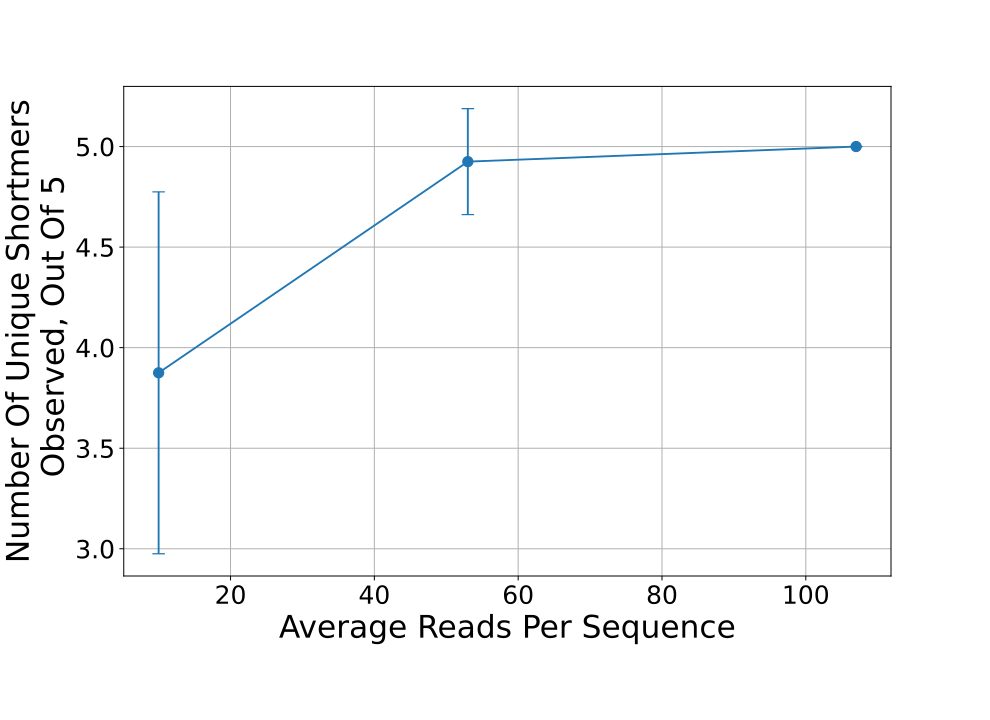}
    \caption{}
    \label{fig:dnastorage-b}
  \end{subfigure}

  \caption{Experiment analysis of Preuss \textit{et al.}~\cite{preuss2024efficient}. (a) Read length distribution, the red dotted line represents the correct sequence length (220nt) (b) Evidence of erasure errors from a downsampling simulation of sequencing reads. The $X$-axis represents the average number of reads per sequence, and the $Y$-axis shows the number of unique shortmers observed out of 5. The simulation was performed with parameters $N = 16$, $K = 5$, and $M = 4$, over 2 different sequences. Error bars indicate the standard deviation over 10 independent down-sampling repeats.
}
  \label{fig:dnastorage}
\end{figure*}

The results from the analysis of the two datasets suggest that a special case of an erasure error is a prominent error while using combinatorial DNA-based data storage systems. This is especially true when using low read counts, as required by larger scale systems. This underscores the need to develop error correction codes, specifically designed to mitigate such erasure errors in combinatorial DNA data storage systems. For a detailed description of the error correction mechanism, see Section~\hyperref[sec:Error-Correcting-Codes-for-Combinatorial-DNA]{Error-Correcting-Codes-for-Combinatorial-DNA}, which shows code construction that is based on mathematical effort by Sabary \textit{et al.}~\cite{sabary2024error}.

\subsection*{Large-Scale Experimental Proof of Concept}
\phantomsection
\label{sec:Experimental-Proof-of-Concept}

Previous research in DNA-based data storage, particularly utilizing combinatorial DNA encodings, has highlighted the necessity of validating the technology and examining the information channels at larger scales. For example, the experimental work of Anavy \textit{et al.}~\cite{anavy2019data}, highlights that large composite DNA alphabets cannot be used in large scale due to the challenges of inferring exact base ratios. This was also demonstrated in the theoretical work that analyzed the error patterns and coverage depth issues in composite and combinatorial DNA-based data storage systems ~\cite{preuss2024sequencing, gruica2024combinatorial, gruica2024geometry, cohen2024optimizing, sokolovskii2024coding}.
Thus, scaling up combinatorial DNA-based data storage systems requires a comprehensive analysis of the information channel to maintain reliability and efficiency.
For that we preformed a large-scale experimental proof of concept study of an end-to-end combinatorial DNA-based storage system. We encoded and decoded a 2.67KB input message using the combinatorial shortmer scheme. Since combinatorial DNA synthesis technology is not yet available, the combinatorial approach was implemented using Bridge oligonucleotide assembly by Yan \textit{et al.}~\cite{yan2023scaling} as an ad-hoc imitation for combinatorial synthesis. 

The experiment constructed 640 combinatorial sequences with similar characteristics to the experiment described in Yan \textit{et al.}~\cite{yan2023scaling} (see Table~\ref{tab:Experiment parameters}). Each combinatorial sequence contained a barcode (referred to as inner barcodes 1-64) and $M=8$ payload cycles over a combinatorial alphabet with \(N=8\) and \(K=4\). The assembly was performed using DNA fragments composed of a 20-mer information sequence and an overlap of 25 bp between adjacent fragments, as shown in Fig.~\ref{fig:large_scale_experiment_640BC_8cycles-a}. 
The 640 combinatorial DNA sequences were constructed in 10 separate pools of 64 sequences each. Each pool was then barcoded separately (referred to as outer barcodes 1-10) and prepared to sequencing using Oxford Nanopore Technologies (See \hyperref[Methods:Molecular-biology-experimental-protocols]{Protocol for large scale proof of concept experiment}). The 10 pools were then sequenced together producing 3,855,233 reads (see Table~\ref{tab:Summary of sequencing reads analyzed in the study}). Only 2,136,534 reads of the output were of length more than 500bp (average = 516, median = 517, sd = 133, see Fig.~\ref{fig:large_scale_experiment_640BC_8cycles-b}). The average number of reads per sequence over all the 640 sequences is 405 with one pool (outer barcode 2) containing only an average of 45 reads per sequence (See Table \ref{tab:Summary of average length of outer BC}).

In this experiment, we incorporated error-correcting codes to enable data recovery even in the presence of errors. Specifically, we utilized a 2D Reed-Solomon (RS) code, as described in Preuss \textit{et al.}~\cite{preuss2024efficient}. Briefly, for each 6-symbol sequence, an inner error correction code (RS(8,6)) was applied adding two redundancy symbols and correcting a single sequence specific error. Additionally, for each block of 30 sequences, an outer code (RS(32,30)) was used adding two redundancy symbols to every sequence of 30 symbols at each given position. 

\begin{table}[ht]
\centering
\begin{tabular}{|l|l|}
\hline
\textbf{Parameters}& \textbf{Value} \\ 
\hline
Text size& 2.67KB\\ 
\hline
Number of sequences& 640\\ 
\hline
Sequence length/number of cycles (M)& 8\\ 
\hline
Logical density (bits/ Synthesis Cycle)&6\\
\hline
Number of unique k-mers (N)& 8\\ 
\hline
Number of unique k-mers in each letter (K)&4\\
\hline
Combinatorial alphabet size & $\binom{8}{4} = 70 \ge 64 = 2^6$\\
\hline
Error correction method& 2D RS\\
\hline
Inner error-correcting code& RS(8,6)\\ 
\hline
Outer error-correcting code& RS(32,30)\\ 
\hline
\end{tabular}
\caption{Design parameters of the large-scale proof of concept experiment.}
\label{tab:Experiment parameters}
\end{table}

\begin{table}[ht]
\centering
\begin{tabular}{|l|l|}
\hline
\textbf{Condition} & \textbf{Value} \\ 
\hline
Number of reads & 3,855,223\\ \hline 
 Number of reads with 500bp+&2,136,524\\ \hline 
 Number of reads with BC identified&260,245\\ \hline
Recovered sequences before error-correction (2D RS)& 631/640\\ 
\hline
Recovered sequences after error-correction (2D RS)& 640/640\\ 
\hline
\end{tabular}
\caption{Summary of the counts of the sequencing reads analyzed in the large-scale proof of concept experiment.}
\label{tab:Summary of sequencing reads analyzed in the study}
\end{table}

\begin{table}[ht]
\centering
\begin{tabular}{|l|l|l|}
\hline
Outer BC& Inner BC &Average number of reads\\ 
\hline
1& 1-64&395\\ \hline 
 2&65-128&\textcolor{red}{45}\\ \hline 
 3&129-192&454\\ \hline 
4& 193-256&646\\ 
\hline
5& 257-320&418\\ 
\hline
6& 321-384&468\\ 
\hline
7& 385-448&406\\ \hline
 8& 449-512&412\\\hline
 9& 513-576&416\\\hline
 10& 577-640&405\\\hline
\end{tabular}
\caption{Summary of the average number of reads per outer BC. Each outer BC was conducted in a different well. Note that the second outer barcode had  fewer reads, and therefore it is highlighted in red.}
\label{tab:Summary of average length of outer BC}
\end{table}

In the original sequence design, every payload sequence position is adjacent to a specific universal sequence (Fig.~\ref{fig:large_scale_experiment_640BC_8cycles-a}). However, our analysis revealed that the sequences contained gaps of various lengths with unidentifiable sequences. This can be seen in Fig.~\ref{fig:large_scale_experiment_640BC_8cycles-b} and Fig.~\ref{fig:large_scale_experiment_640BC_8cycles-c}, where the overall sequence length exceeds the expected length and the position in which each universal sequence is found within the sequence varies significantly. These observations suggest that the naive BLAST-based algorithm presented before required some adjustments to support longer combinatorial sequences containing multiple cycles. We thus adjusted the BLAST-based decoding algorithm to support the identification and analysis of multiple universal-payload sequence pairs in every analyzed read. The algorithm was also adjusted to support analysis of reads containing only a subset of the Universal-payload pairs. Full description of the adjusted algorithm is available in the Methods section ~\hyperref[Methods:640BC-8cycles]{BLAST-based analysis of large-scale combinatorial experiment}.

  \begin{figure*}[htbp]
  \centering
  \begin{subfigure}[b]{0.90\linewidth}
    \centering
    \includegraphics[width=\linewidth]{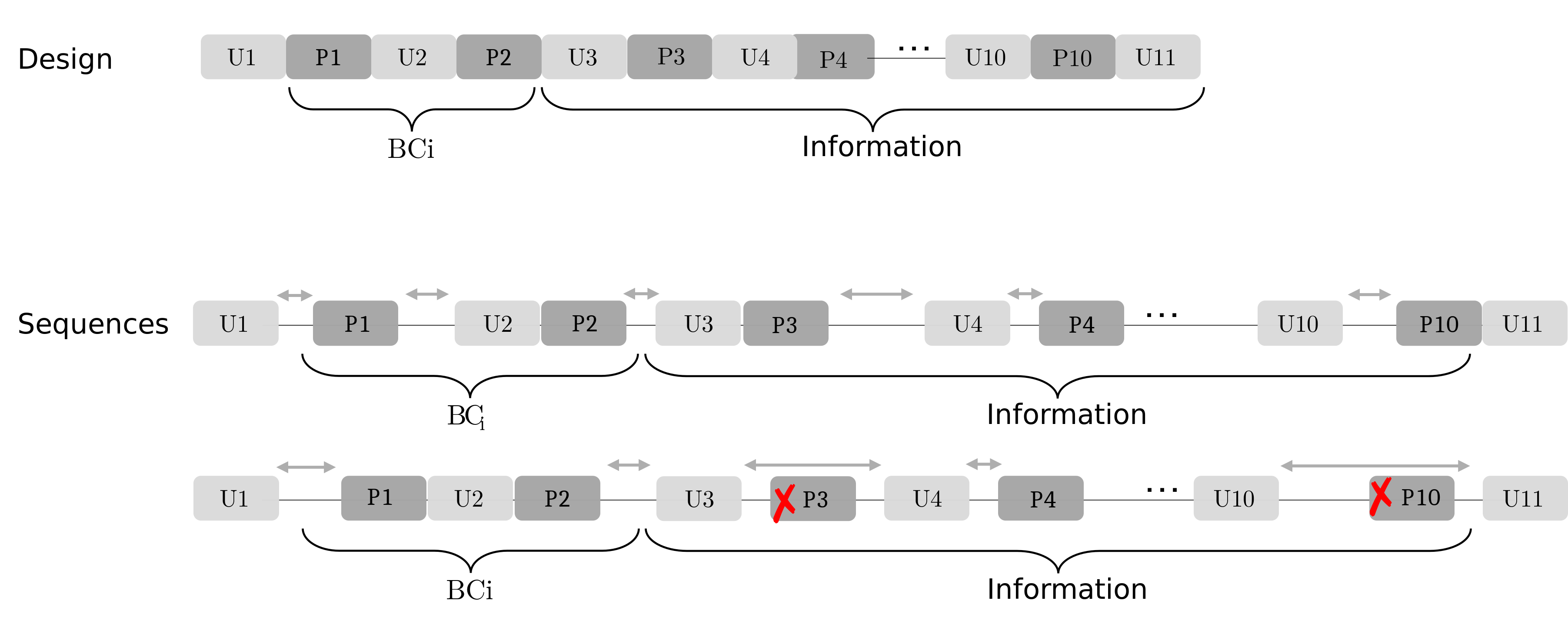}
    \caption{}
    \label{fig:large_scale_experiment_640BC_8cycles-a}
  \end{subfigure}
  
  \begin{subfigure}[b]{0.43\linewidth}
    \centering
    \includegraphics[width=\linewidth]{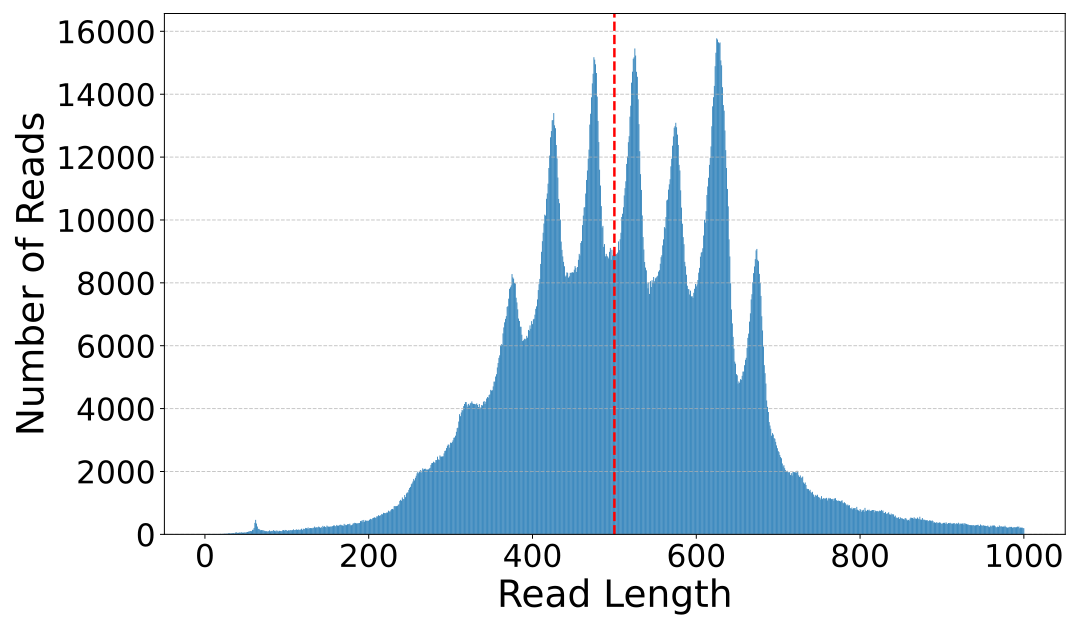}
    \caption{}
    \label{fig:large_scale_experiment_640BC_8cycles-b}
  \end{subfigure}
  \hfill
  \begin{subfigure}[b]{0.56\linewidth}
    \centering
    \includegraphics[width=\linewidth]{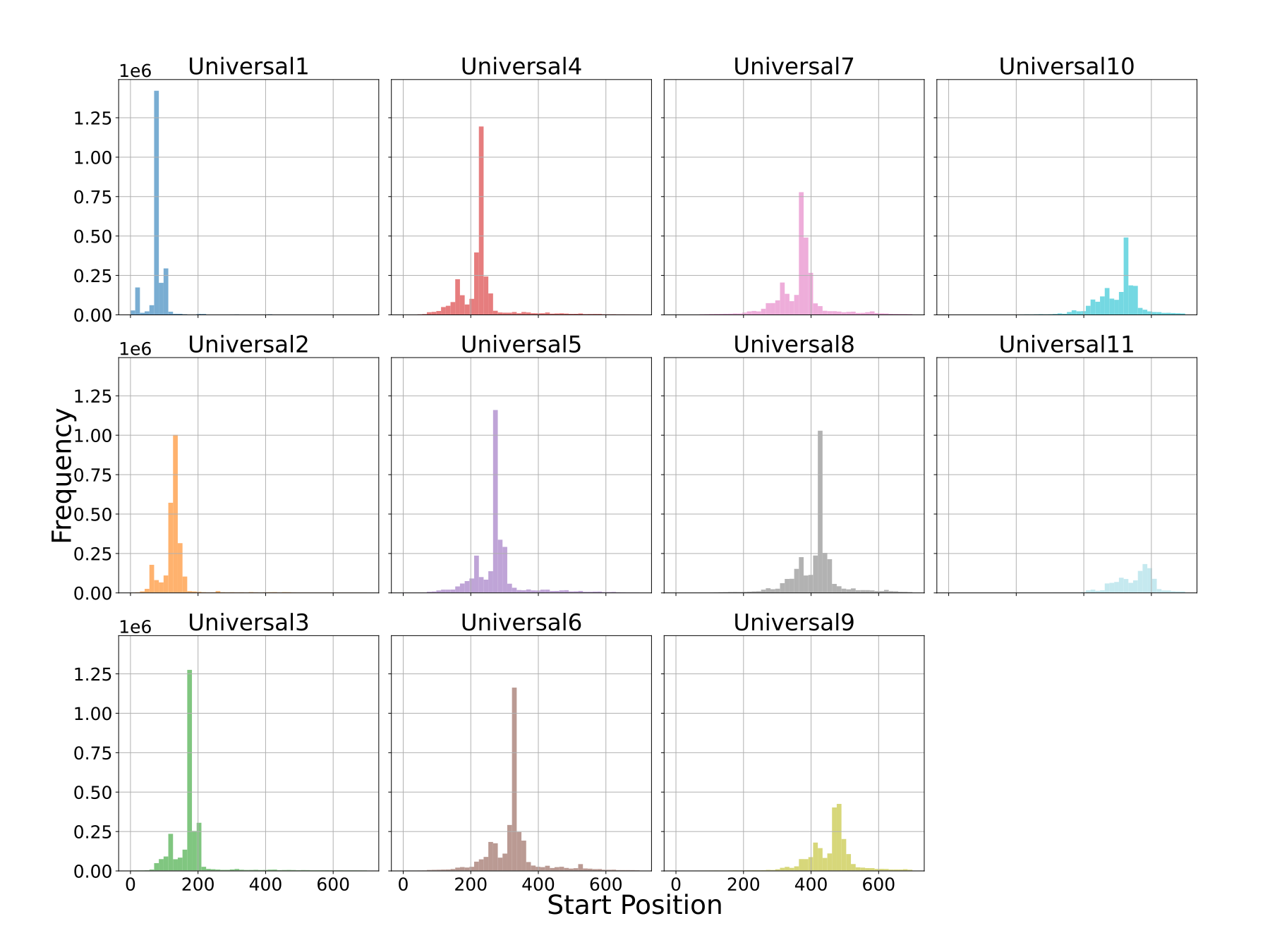}
    \caption{}
    \label{fig:large_scale_experiment_640BC_8cycles-c}
  \end{subfigure}

  \caption{Analysis of data from large-scale experimental proof of concept. (a) Sequence design and potential errors. The top diagram presents a schematic description of the sequence design, including universals and payloads. The bottom two diagrams show two potential erroneous reads, in which the position of the universal sequences and payload varies. The red cross corresponds to missing or unidentifiable payloads.  (b) Read length distribution, (c) All universals start position in the sequences. }
  \label{fig:large_scale_experiment_640BC_8cycles}
\end{figure*}

The analysis of 100\% of the data revealed that 631 sequences were successfully recovered out of a total of 640 sequences. For each of these 631 sequences, all four payload sequences across eight cycles were fully reconstructed.  Our 2D RS error-correcting approach  effectively corrected the errors in the 9 unsuccessful sequences, allowing us to successfully recover all 640 sequences and the original encoded message. 

Although these 9 sequences were ultimately corrected by the 2D RS code, we conducted a deeper analysis of their error profiles to gain a better understanding of the underlying error behavior.
 The primary issue appears to be the low read count in the sequence pool that were barcoded with outer barcode BC2. This limited read coverage significantly impacted the recovery of these 9 sequences. A characterization of the different  error types is described in the histograms presented in Fig. \ref{fig:large_scale_experiment_640BC_8cycles-zoom-in}. In the histogram presented in Fig.~\ref{fig:large_scale_experiment_640BC_8cycles-tie-payload-a}, we can see that in the seventh cycle of inner barcode BC70 we observed more than four payloads, with a tie on the fourth most common payload sequence. We therefore had to select which payload sequence to retain. This approach occasionally resulted in selecting the incorrect set of four payloads, which ultimately led to the wrong letter, $\sigma$, being chosen. In the histogram presented in Fig.~\ref{fig:large_scale_experiment_640BC_8cycles-3-payloads-b}, we can see that in the seventh cycle of inner barcode BC75, only three payload sequences are observed, resulting in an erasure error. As seen in Fig.~\ref{fig:large_scale_experiment_640BC_8cycles-observed-more-incorrecct-payload-c}, the histogram shows that the second cycle of inner barcode BC583, there were three payloads with counts exceeding 75, while the fourth payload has a significantly lower count. Notably, the count for payload 6 (red square) is higher than that of payload 4 (green square), even though payload 4 represents the correct payload. To avoid accidental mixups, we sometime require to observe $t>1$ copies of each payload sequence. This may result in erasure errors where only three payload sequences are observed with enough copies. Choosing the correct threshold $t$ is an important design parameter for the decoding process representing the tradeoffs between different error types. Similar read-count histograms for all 9 unsuccessful sequence is available in Supplementary~\hyperref[sup:Experimental-proof-of-concept-9BC-that-did-not-recover-out-of-the-640BC]{Experimental proof of concept 9BC that did not recover out of the 640BC}.

We next performed a subsampling analysis in which we analyzed the results for a random subset of the sequencing output. We see that, with an average of $\sim100$ reads or less per sequence we encounter similar erasure errors more frequently, see Fig.~\ref{fig:large_scale_experiment_640BC_8cycles-d}.

\begin{figure*}[htbp]
  \centering
  \begin{subfigure}[b]{0.48\linewidth}
    \centering
    \includegraphics[width=\linewidth]{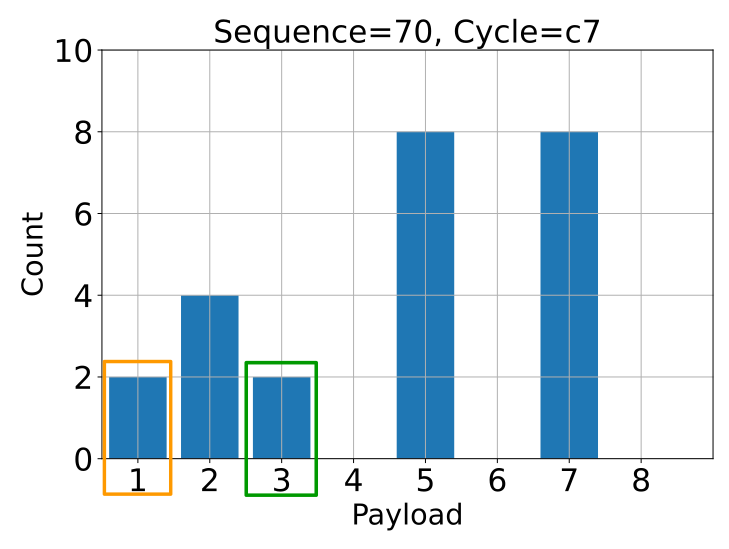}
    \caption{}
    \label{fig:large_scale_experiment_640BC_8cycles-tie-payload-a}
  \end{subfigure}
  \hfill
  \begin{subfigure}[b]{0.48\linewidth}
    \centering
    \includegraphics[width=\linewidth]{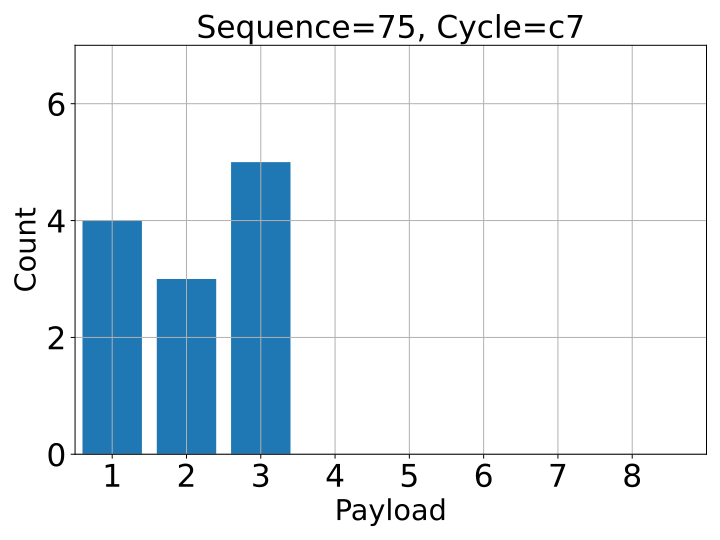}
    \caption{}
    \label{fig:large_scale_experiment_640BC_8cycles-3-payloads-b}
  \end{subfigure}
  
  \vspace{0.5cm} 
  \begin{subfigure}[b]{0.48\linewidth}
    \centering
    \includegraphics[width=\linewidth]{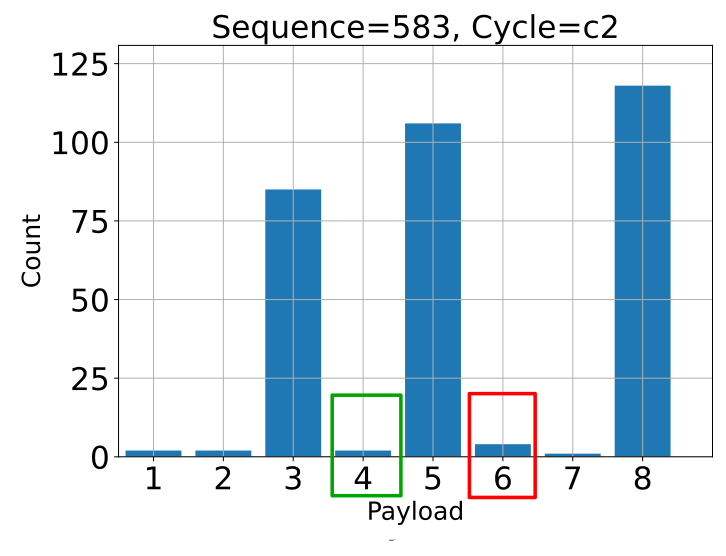}
    \caption{}
    \label{fig:large_scale_experiment_640BC_8cycles-observed-more-incorrecct-payload-c}
  \end{subfigure}
  \hfill
  \begin{subfigure}[b]{0.48\linewidth}
    \centering
    \includegraphics[width=\linewidth]{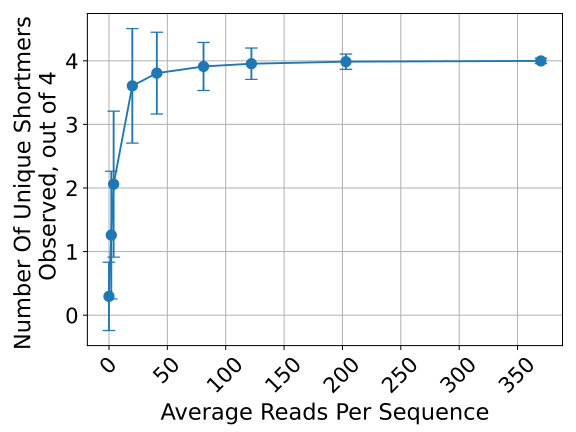}
    \caption{}
    \label{fig:large_scale_experiment_640BC_8cycles-d}
  \end{subfigure}

  \caption{Experiment analysis of our first large-scale datasets. This dataset include 640 combinatorial sequences and the parameters are $N = 8$, $K = 4$, and $M = 8$.  (a) Payload count of sequence BC70, cycle 7, showing a tie where payload 1 (orange square) was incorrectly chosen over payload 3 (green square). (b) Payload count of sequence BC75, cycle 7, illustrating a case where only three payloads were recovered instead of four. (c) Payload count of sequence BC583, cycle 2, showing three payloads with counts exceeding 75 with the forth having very few reads. (d) Subsampling analysis of the sequencing reads. The $X$-axis represents the average number of reads per sequence, and the $Y$-axis shows the number of unique shortmers observed out of 4, average is taken over all barcodes and all cycles. Error bars represent the standard deviation across 10 independent subsampling repeats for each read depth.}
  \label{fig:large_scale_experiment_640BC_8cycles-zoom-in}
\end{figure*}

\subsection*{Error-Correcting Codes for Combinatorial DNA}
\phantomsection
\label{sec:Error-Correcting-Codes-for-Combinatorial-DNA}

Our analysis of previously published datasets and our experimental PoC datasets (see Figure~\ref{fig:HellowWorld}, Figure~\ref{fig:dnastorage} and Figure~\ref{fig:large_scale_experiment_640BC_8cycles-zoom-in}) shows a high occurrence of composite erasures, in which one (or more) of the shortmers composing a combinatorial letter is absent from the sequencing reads. Such an event is more likely to occur when the coverage depth is relatively low. To address these specific errors in combinatorial DNA, based on the construction proposed by Sabary et al.~\cite{sabary2024error}, we propose an error-correction scheme,  which we name as the \emph{Combinatorial VT code}, or shortly \emph{Combi VT code}. A thorough analysis of the redundancy properties and a formal mathematical description of the Combi-VT code construction are provided in Sabary et al.~\cite{sabary2024error}. In that work, the composite asymmetric erasure error model is studied from a coding-theoretic perspective, an extension of Hamming distance is used as a distance metric for such code, and explicit bounds on the redundancy and achievable information rate (information density) are derived. In particular, it is shown that in the regime $N = M$, the redundancy of the proposed construction is near-optimal, differing from the theoretical optimum by an additive term of less than $\log N$. The same work also introduces several generalizations of the basic construction that enable correction of multiple asymmetric erasures within a single combinatorial symbol.  Generally speaking, and as demonstrated both in simulations and in real datasets in the following two sections, assuming similar combinatorial alphabets, since this scheme is specifically designed to address composite asymmetric erasures, it can correct a greater number of such errors while using the same redundancy as the 2D RS code.

In this scheme, sequences over the combinatorial DNA alphabet are modeled as binary matrices, where rows represent the combinatorial letters in their order within the sequence, and columns represent the building block shortmers. Each row contains exactly $K$ ones, indicating the set of $K$ building block shortmers that compose the combinatorial letter (see Fig.~\ref{fig:combEnc}d). In this model,  missing shortmers (composite erasures) are treated as \emph{asymmetric errors}, where bits in the matrix are flipped from $1$ to $0$, but not vice versa (see Fig.~\ref{fig:combDec}b). The scheme leverages tensor-product (TP) codes~\cite{wolf2006introduction} to combine standard erasure and substitution-correcting codes (e.g., Reed-Solomon (RS) codes~\cite{reed1960polynomial}) with Varshamov-Tenengolts (VT) codes~\cite{varshamov1965code}, which are designed to correct asymmetric errors. Using this construction, it is possible to correct typical error patterns involving missing shortmers and to reliably decode information stored in combinatorial DNA molecules.

Our proposed error-correcting code is specifically tailored to handle combinatorial asymmetric errors and is based on the following idea. Since each combinatorial letter is a set of $K$ $k$-mers, the code partitions all possible such sets according to a mathematical expression known as the Varshamov-Tenengolts (VT) syndrome. Each set of $K$ $k$-mers, can be mapped to a length-$N$ binary vector with exactly $K$ ones. The $N$ entries of the vectors corresponds to the $N$ building blocks $k$-mers, and $0$ and $1$ corresponds to the presence or absence (respectively) of a $k$-mers in the set. Thus, the VT syndrome of each set can be uniquely calculated using the formula which is given as follows. Given a length-$N$ binary vector $\mathbf{x} = (x_0, x_1, \ldots, x_{N-1}) \in \{0,1\}^N$, and a positive integer $p$, we define the \emph{VT-syndrome over $p$}~\cite{varshamov1965code} of $\mathbf{x}$, denoted by $\mathbf{s}_{}^p(\mathbf{x})$, as follows 
$\mathbf{s}_{}^p(\mathbf{x}) \triangleq \sum_{i=0}^{N-1} i x_i \bmod p.$ Note that we have that, $\mathbf{s}_{}^p(X) \in \{0, 1, \ldots, p-1 \}$. For the simplicity of our construction, we define $p$ as the smallest prime power, such that $p\ge N$. 

As shown by Varshamov and Tenengolts~\cite{varshamov1965code}, this syndrome enables correction of any single asymmetric error. Thus, as depicted in Fig~\ref{fig:combDec}c and Fig.~\ref{fig:combDec}d, if only $K-1$ out of $K$ $k$-mers are observed for a combinatorial letter, it is possible to identify the missing $k$-mer by comparing the VT syndrome of the observed $K-1$ symbols to the VT syndrome of the original transmitted word and thereby recover the full combinatorial letter. 

In our scheme, the VT syndromes of the combinatorial letters in each sequence are protected against errors. This is done by encoding the vector of VT syndromes using an RS code (see Fig.~\ref{fig:combEnc}b). Thus, in our setting, it is possible to recover up to two VT syndromes per sequence, and then use each recovered syndrome to identify and correct the corresponding missing shortmer. The construction is flexible and can be extended to scenarios with a larger number of errors. In~\cite{sabary2024error}, it is described how to adapt the construction to handle multiple missing shortmers in multiple combinatorial letter by increasing the redundancy of the RS code and using extensions of the VT syndromes, thereby enabling the correction of more than two erroneous letters in each sequence. 

\textbf{Encoding.} Based on this idea, the encoding process begins by mapping most of the user's binary information into valid combinatorial letters, each defined by a set of $k$ shortmers (as described in Fig.~\ref{fig:view of a combinatorial alphabet} and Fig.~\ref{fig:combEnc}a, note that in the latter, the matrix is shown as the transpose of the one in the former). For each such a letter, our encoder calculates its VT syndrome (see Fig.~\ref{fig:combEnc}b). This mapping associates each letter with a unique VT syndrome value and organizes the space into disjoint classes (known as \emph{cosets}) according to the syndrome. The vector of VT syndromes is then encoded using an RS code, which provides protection against errors and allows to recover them in case of errors (see Fig.\ref{fig:combEnc}b). Next, the remaining information bits from the user, together with the RS-encoded VT syndromes, are mapped to additional combinatorial letters (see Fig.\ref{fig:combEnc}c). Finally, the complete combinatorial sequence is sent for synthesis (see Fig.\ref{fig:combEnc}d).

To support both structure and error correction, we define two types of encoding functions. The first maps a binary message directly to a combinatorial letter. The second maps a shorter binary message together with a VT syndrome, which will later allow the decoder to identify and correct missing shortmers. This two-stage design provides a compact representation while ensuring robustness to asymmetric errors. The full encoder of this construction is given in Algorithm 1 of Sabary et al.~\cite{sabary2024error}.

\textbf{Decoding.} As described above, any shortmer that is missing because it was not identified or sequenced appears as an asymmetric error in the matrix representation of the encoded combinatorial sequence. For example, in Fig.\ref{fig:combDec}a, two entries in the matrix with value 1 have flipped to 0.

The decoding process operates in two stages. First, the VT syndromes of the combinatorial letters are computed (see Fig.~\ref{fig:combDec}b). Since each VT syndrome is protected with an RS code, it can be recovered even when shortmers are missing (see Fig.~\ref{fig:combDec}c). Second, the decoder uses the $K-1$ observed shortmers and their associated VT syndrome to recover the missing shortmer, enabling complete and accurate decoding (see Fig.~\ref{fig:combDec}d). Once the matrix is fully recovered, the original binary information can be decoded through the predefined mapping. A full description of the decoding algorithm is provided in Sabary et al.~\cite{sabary2024error}, and an implementation of this decoding is available in our code.

It is important to note that this construction is applicable only under specific code parameters. In particular, it requires the combinatorial letters to be approximately uniformly distributed across the VT syndrome classes. This uniformity ensures that all syndrome classes are populated and usable for encoding, which is essential for enabling decoding from partial observations. A mathematical description of this constraint, along with a proof of the construction, is provided in ~\cite{sabary2024error}. A schematic description of the encoding process is described in Fig.~\ref{fig:V2 RS(7,5) Flowchart - Encoding Process of Combinatoric Payload Erasure}. 

\begin{figure}
    \centering
    \includegraphics[width=\linewidth]{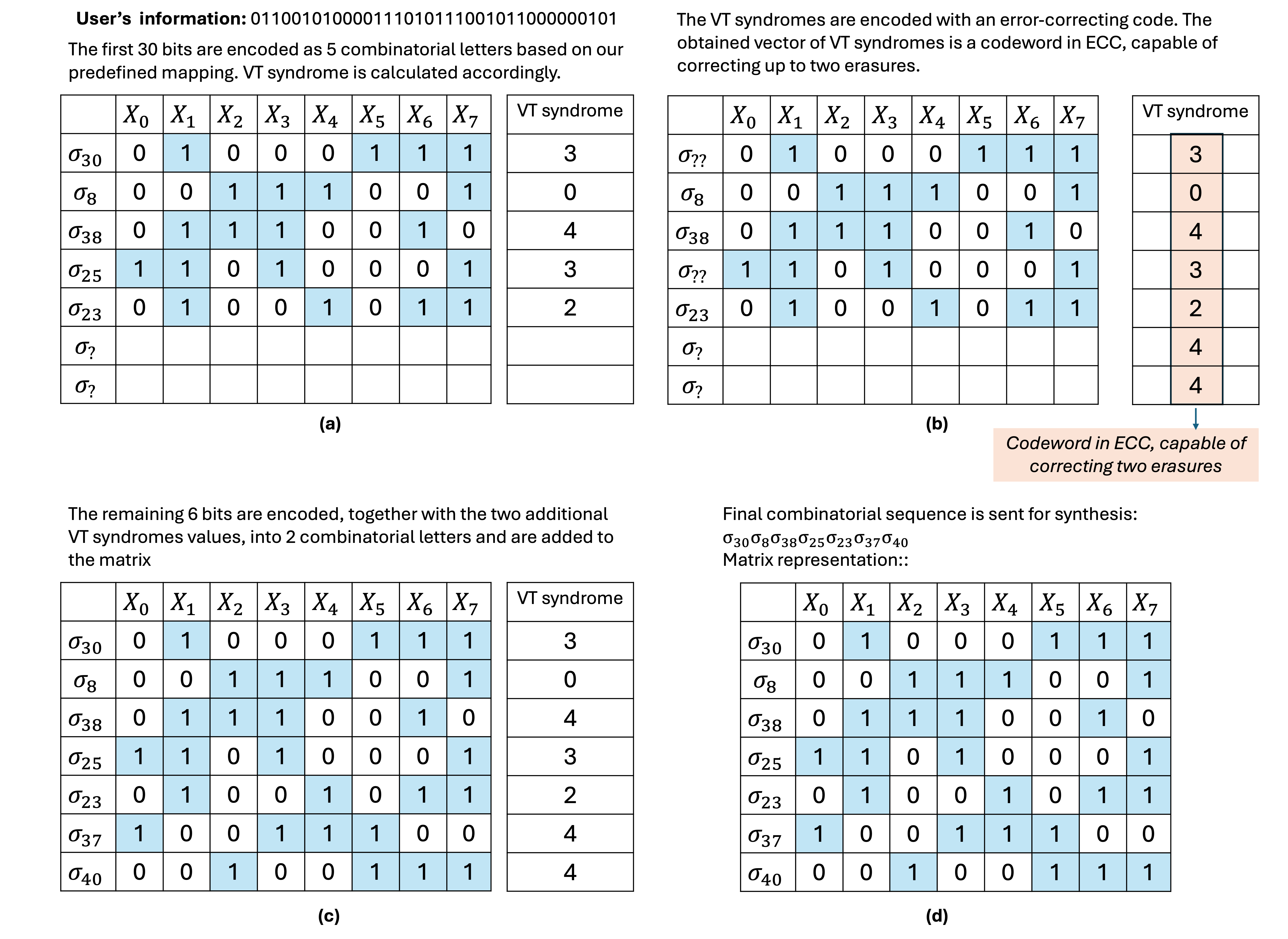}
    \caption{Example of the encoding process for combinatorial sequences using VT syndromes and RS coding, with parameters $N=8$, $K=4$, $M=7$, $p=2^3=8$. (a) The user’s binary information is first mapped into valid combinatorial letters, each represented as a column in the matrix and defined by a set of  $K=4$ shortmers (blue squares with value 1).(b) For each combinatorial letter, the VT syndrome is computed. The resulting sequence of syndromes forms a codeword in a  Reed–Solomon (RS) error-correcting code, namely RS$(7,5)$ over $GF(2^3)$, this encoding can correct up to two erasures. (c) The RS-encoded VT syndromes are mapped into additional combinatorial letters and appended to the original information-carrying letters, forming the complete encoded matrix. (d) The final combinatorial sequence is obtained from the complete matrix and prepared for synthesis. For more details, analysis, and generalizations of this code construction, the reader is referred to Sabary et al.~\cite{sabary2024error}.}
    \label{fig:combEnc}
\end{figure}

\begin{figure}
    \centering
    \includegraphics[width=\linewidth]{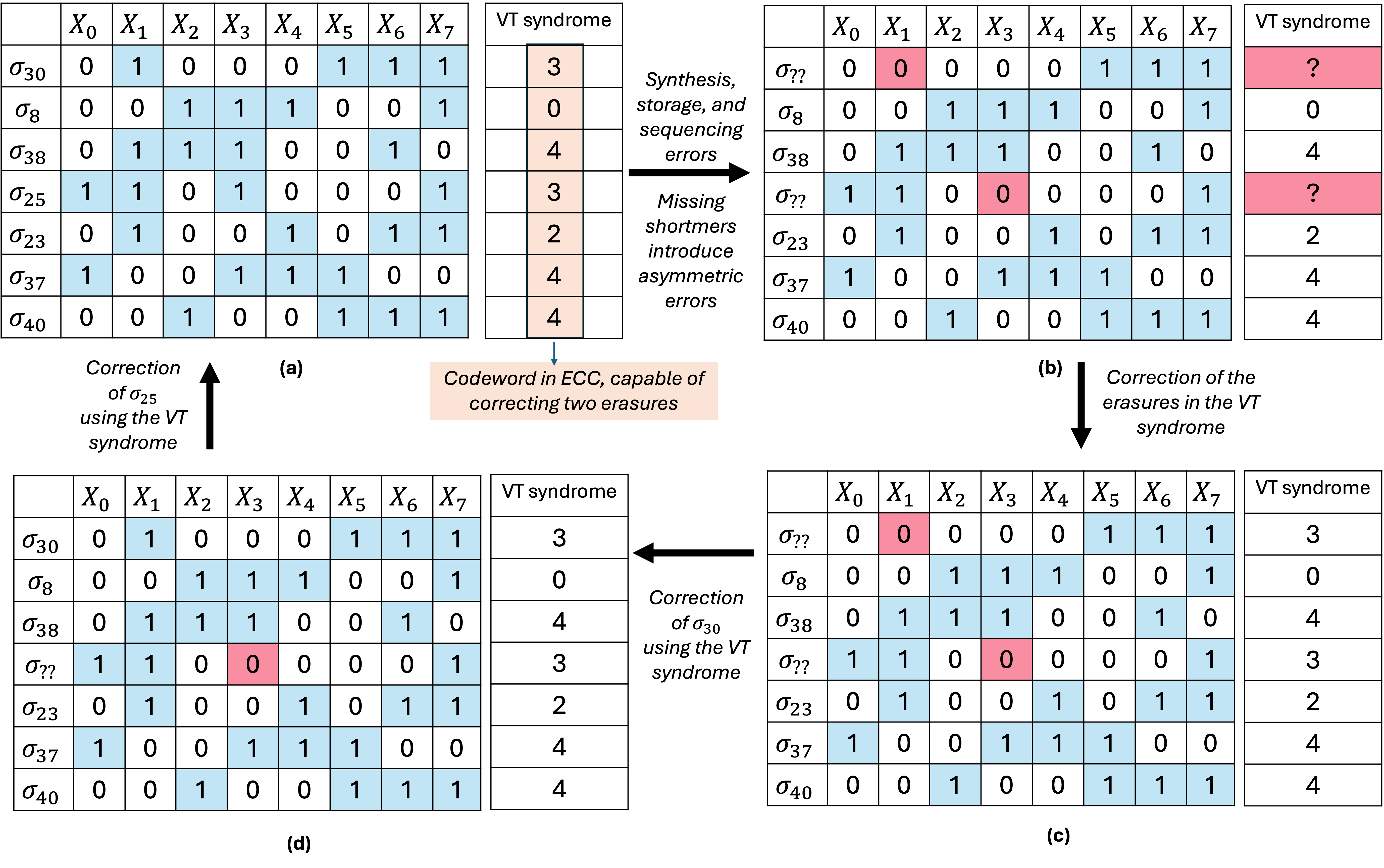}
    \caption{Example of the decoding process for combinatorial sequences using VT syndromes and RS coding, with parameters $N=8$, $K=4$, $M=7$, $p=2^3=8$. (a) Matrix representation of the received combinatorial sequence after sequencing. Missing shortmers, shown as red 0s, correspond to asymmetric errors (1s flipped to 0s). (b) For each combinatorial letter (matrix column), the VT syndrome is computed from the observed shortmers. Missing shortmers cause the corresponding syndrome values to be unknown (marked with “?”). The sequence of VT syndromes was designed to be a codeword in a RS error-correcting code, and is capable correcting two erasures. (c) RS decoding is applied on  to recover the missing VT syndromes, even when shortmers are absent. (d) Using each recovered VT syndrome together with the correspond $k-1$ observed shortmers in its row, the decoder identifies and completes the missing shortmers, and thus the matrix is recovered.}
    \label{fig:combDec}
\end{figure}



\begin{figure}[ht]
\centering
\includegraphics[width=0.5\textwidth]{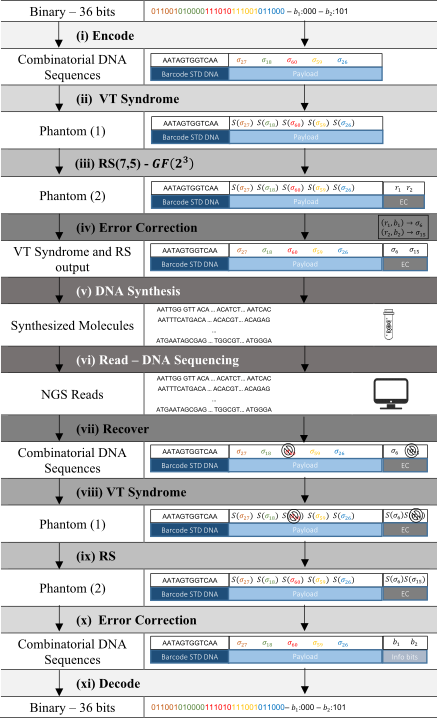}
\caption{End-to-end workflow of a combinatorial DNA-based storage system. A binary message is broken into chunks of information bits, these chunks are barcoded, and encoded into a combinatorial alphabet, while some of the infromation bits in each chunk are designated for redundancy \(b_i\) (i). On each chunk, VT syndrome calculations are performed on each combinatorial letter (ii). An RS encoding \(r_i\) is added to each VT syndrome output chunk (used as the syndrome) (iii). The designated information bits \(b_i\) are added to every \(r_i\) in each chunk (used as the syndrome index), and these (\(r_i\),\(b_i\)) bits are then converted into combinatorial letters (iv). The combinatorial message is synthesized using combinatorial shormer synthesis (v), and the DNA is sequenced (vi). Next, the combinatorial letters are reconstructed (vii), and the VT syndrome is performed on the combinatorial letters (viii). Using RS on each chunk, the letters are decoded (ix), and finally, the combionatorial letters and  information bits are recovered (x), followed by translating the infromation back into the original binary message (xi).}
\label{fig:V2 RS(7,5) Flowchart - Encoding Process of Combinatoric Payload Erasure}
\end{figure}

We note that a direct comparison between the proposed Combi-VT coding scheme and previously published DNA-based data storage coding schemes is not straightforward. First, combinatorial DNA relies on a fundamentally different data-representation paradigm compared to conventional base-by-base DNA synthesis. Second, the underlying error models differ substantially: while most prior works focus on substitution, insertion, and deletion errors at the nucleotide level, the present work addresses asymmetric erasures at the shortmer level, which arise naturally in combinatorial synthesis. For example, the DNA Fountain scheme~\cite{erlich2017dna} assumes that a sufficient fraction of sequencing reads are effectively noise-free, whereas Antkowiak et al.~\cite{antkowiak2020low} proposed a Reed–Solomon–based approach designed to tolerate higher synthesis error rates in low-cost manufacturing settings. Chandak et al.~\cite{chandak2019improved} introduced an LDPC-based coding framework integrated directly with the Nanopore basecaller, targeting a different error profile and system architecture. Combinatorial synthesis is a relatively recent approach to DNA data representation and exhibits a distinct and asymmetric error structure compared to prior synthesis technologies. Accordingly, we do not claim that the proposed coding scheme is universally superior to existing DNA storage codes. Rather, our objective is to design and evaluate a coding strategy that is well-matched to the dominant error mechanisms inherent to combinatorial DNA synthesis. 
\subsection*{Comparison of the Combinatorial VT Code and the 2D RS Code Using Simulations}
\phantomsection
\label{Simulation-of-the-New-Error-Correction-vs-2D-RS}

We compared the error correction capabilities of the combi VT scheme and the 2D RS scheme using an end-to-end simulation, similar to the one used in Preuss \textit{et al.}~\cite{preuss2024efficient}. Briefly, 10KB messages were encoded using a combinatorial alphabet with $N=8$ and $K=4$ corresponding to 6 bits per symbol. Each message was encoded using 2,432 combinatorial sequences of length 7 (including EC redundancy symbols and sequences). The parameters of the two inner EC codes used the same redundancy level of 6 redundancy bits for every 30 bits of information. The 2D RS EC used an RS(7,6) inner code in each sequence, and the combi VT code uses an RS(7,5) EC to encode the VT syndromes vector. The outer EC code used RS(32,30), adding two redundancy sequences for each 30-sequence block. A simulation of combinatorial synthesis, molecule sampling and sequencing was then done on the combinatorial sequences using tunable parameters for insertion, deletion, and substitution error rates and for the sampling rate (see Methods~\hyperref[Methods:Simulation-of-combinatorial-DNA-based-storage-system]{Simulation of a combinatorial DNA-based storage system}). The results of the simulation runs are summarized in Fig.~\ref{fig:Simulation-of-an-end-to-end-combinatorial-shortmer-encoding}. Fig.~\ref{fig:Simulation-of-an-end-to-end-combinatorial-shortmer-encoding-a} depicts a schematic representation of the simulation workflow and indicates the reconstruction rates and Levenshtein distance calculations throughout the process .
 Each run included 30 repeats with random input texts of 10 KB. In Fig.\ref{fig:Simulation-of-an-end-to-end-combinatorial-shortmer-encoding-b}, an average of 20 reads per sequence were sampled, demonstrating the effects of both error correction codes. The results clearly show that the combi VT code achieves better reconstruction outcomes compared to the 2D RS method~\cite{preuss2024efficient} with a significant average improvements of 0.012 (P-value = 3.82e-11,two-sided t-test), and 0.005 (P-value = 7.53e-8,two-sided t-test) in the normalized Byte-level Levenshtein distance score for After inner EC and After outer RS, respectively. Furthermore, as shown in Fig.\ref{fig:Simulation-of-an-end-to-end-combinatorial-shortmer-encoding-c}, when using 10, 15, 20, and 30 sampled reads per sequence, the new error correction code for the combinatorial encoding scheme\cite{sabary2024error} consistently provides better results and more accurate data reconstruction. Notably, the sampling rate is a dominant factor affecting the reconstruction rate. Even with no synthesis and sequencing errors, low sampling rates yield poor results (Fig.~\ref{fig:Simulation-of-an-end-to-end-combinatorial-shortmer-encoding-c}) that the error correction cannot overcome. The effect of substitution errors on overall performance is comparatively smaller and easier to detect and correct. This is because substitution errors occur at the nucleotide level while insertion and deletion errors occur at the $k$-mer level. The minimal Hamming distance $d=2$ of the trimer set $\Omega$ allows for the correction of single-base substitutions.

\begin{figure*}[htbp]
  \centering
  \begin{subfigure}[b]{0.48\linewidth}
    \centering
    \includegraphics[width=\linewidth]{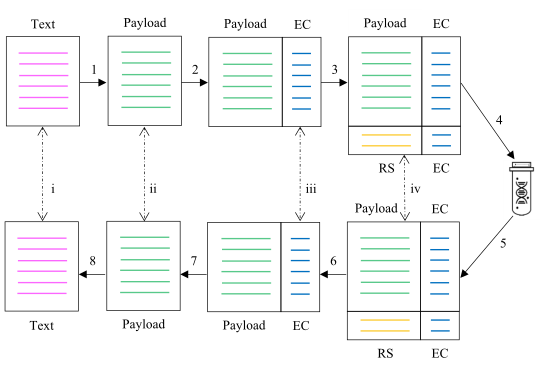}
    \caption{}
    \label{fig:Simulation-of-an-end-to-end-combinatorial-shortmer-encoding-a}
  \end{subfigure}
  \hfill
  \begin{subfigure}[b]{0.48\linewidth}
    \centering
    \includegraphics[width=\linewidth]{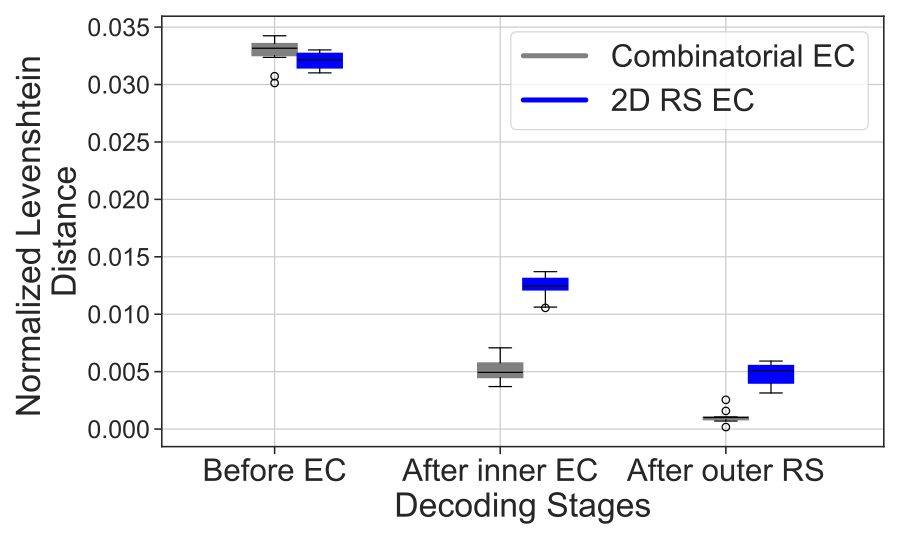}
    
    \caption{}
    \label{fig:Simulation-of-an-end-to-end-combinatorial-shortmer-encoding-b}
  \end{subfigure}
  
  \vspace{0.5cm} 
  \begin{subfigure}[b]{0.80\linewidth}
    \centering
    \includegraphics[width=\linewidth]{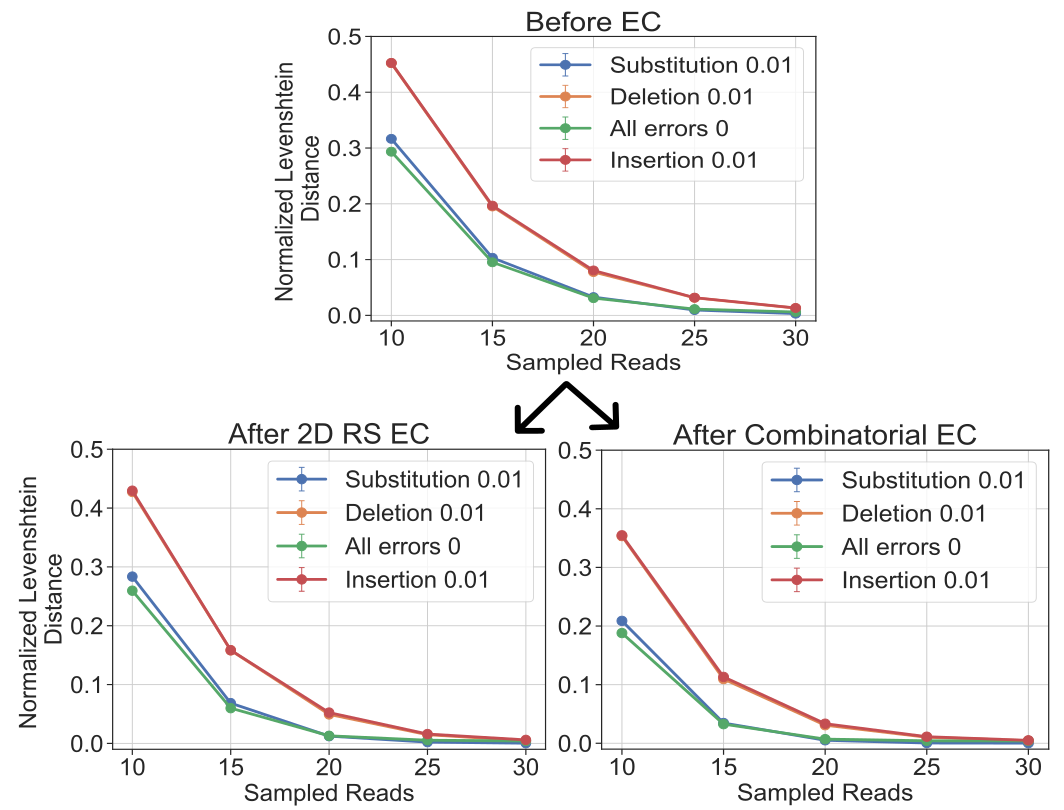}
    \caption{}
    \label{fig:Simulation-of-an-end-to-end-combinatorial-shortmer-encoding-c}
  \end{subfigure}

  \caption{Comparison of error correction schemes using end-to-end simulations. (a) A schematic view of the simulation workflow including: encoding into a combinatorial message (1), Inner Error correction encoding (2), Outer Error correction encoding (3), DNA synthesis and sequencing simulation (4), combinatorial letter reconstruction (5) Outer Error correction decoding (6). Inner Error correction decoding (7), and decoding of the original text message (8). The Roman numerals (i-iv) represent the different performance evaluation steps. (b) Error rates observed in different stages of the decoding process. Boxplot (representing 30 simulations) of the normalized Byte-level Levenshtein distance (See Methods~\hyperref[Methods:Simulation-Reconstruction]{Simulation Reconstruction}) in the different stages of the simulation. Simulation was done with a coverage of 20 reads per barcode, with a substitution error rate of 0.01 for Combi VT code (grey) and 2D RS EC (blue). (c) Sampling rate effect on overall performance. Average normalized Byte-level Levenshtein distance (over 30 simulations) as a function of sampling rate before EC decoding, and after decoding (ii). Different lines represent different error types (substitution, deletion, and insertion) introduced at a rate of 0.01 and 0.}
  \label{fig:Simulation-of-an-end-to-end-combinatorial-shortmer-encoding}
\end{figure*}

\subsection*{Comparison of the Combinatorial VT Code and the 2D RS Using Actual Experimental Data} \phantomsection
\label{Experiment of 320BC 2D RS vs 320BC new error correction}

To evaluate the effectiveness of our proposed error correction scheme, we conducted a second large-scale experiment using the combinatorial Bridge oligonucleotide assembly by Helixworks Technologies Ltd. The experiment compared the performance of the traditional 2D RS-based approach with our newly designed combinatorial error correction method. The experimental parameters where slightly modified to allow a fair comparison of the two approaches using identical redundancy levels. Briefly, we encoded and decoded a 2.66KB input message using 640 combinatorial sequences, with $M=7$ cycles, $N=8$ and $K=4$ (see Table~\ref{tab:Experiment parameters 320bc}). see Section \hyperref[sec:Experimental-Proof-of-Concept]{First Experimental Proof of Concept 640 Sequences and 8 Cycles} for more details.



The sequencing output 3,615,438 contained
 reads (see Table~\ref{tab:Summary of sequencing reads analyzed in the study 320bc}) out of which 1,590,536 reads were of length more than 450bp (average = 420, median = 419, sd = 141, see Fig.~\ref{fig:large_scale_experiment_640BC_7cycles-a}). The average number of reads per sequence over all the 640 sequences is 621. (See Table \ref{tab:Summary of average length of outer BC 320bc} and Fig.\ref{fig:large_scale_experiment_640BC_7cycles-c}).

\begin{figure*}[htbp]
 
  \begin{subfigure}[b]{0.43\linewidth}
    \centering
    \includegraphics[width=\linewidth]{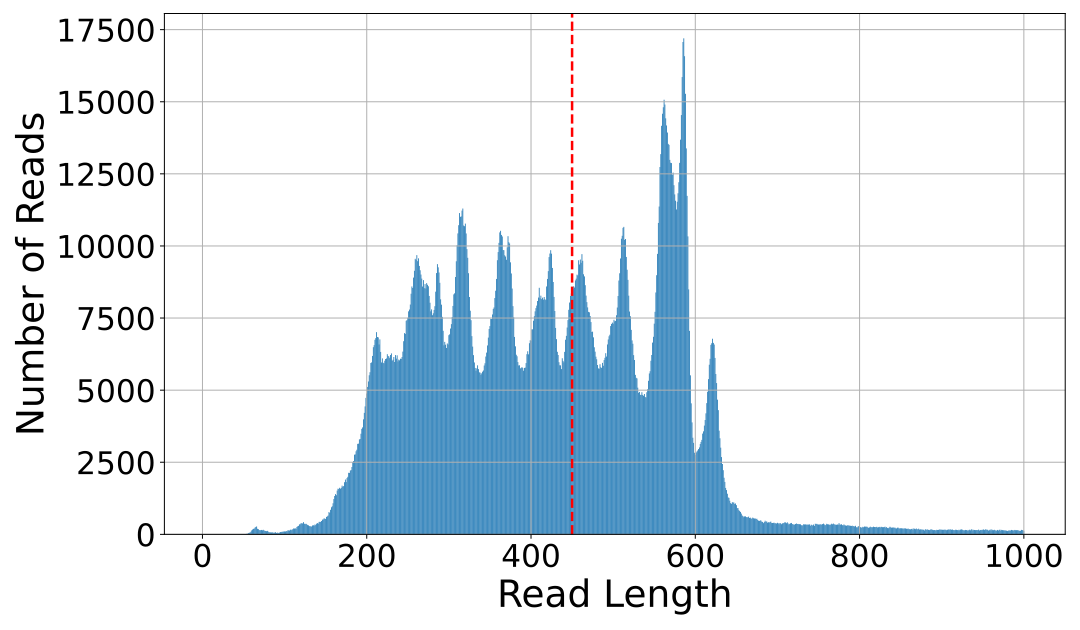}
    \caption{}
    \label{fig:large_scale_experiment_640BC_7cycles-a}
  \end{subfigure}
  \hfill
  \begin{subfigure}[b]{0.56\linewidth}
    \centering
    \includegraphics[width=\linewidth]{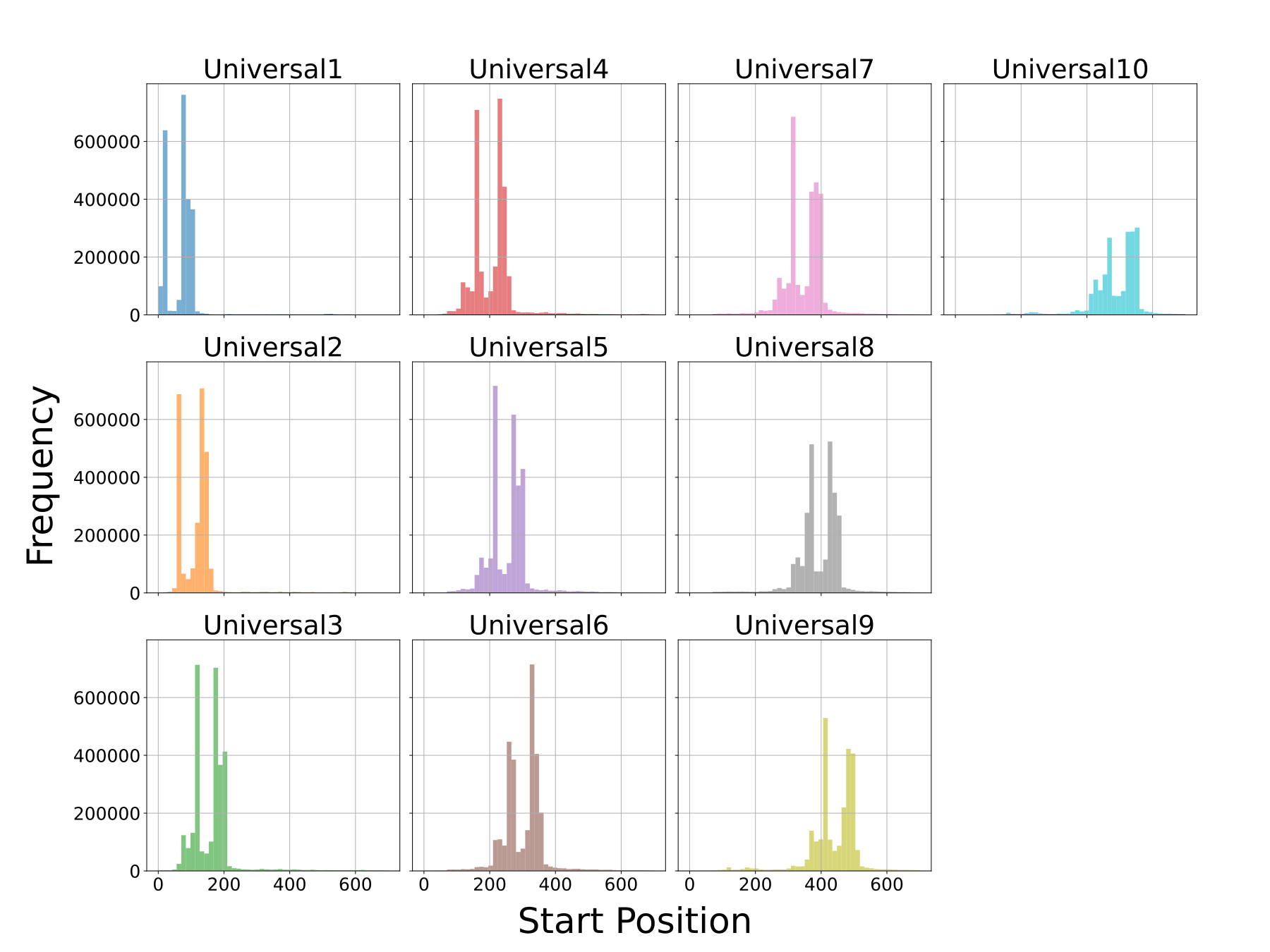}
    \caption{}
    \label{fig:large_scale_experiment_640BC_7cycles-b}
  \end{subfigure}

    \begin{subfigure}[b]{0.43\linewidth}
        \centering
        \includegraphics[width=\linewidth]{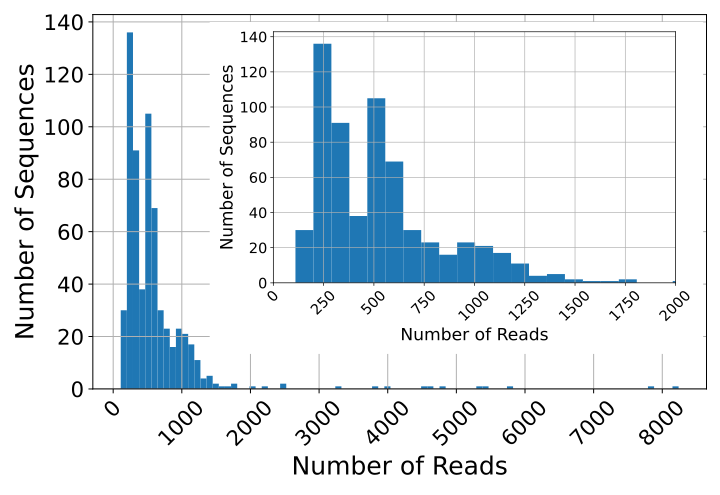}
        \caption{}
        \label{fig:large_scale_experiment_640BC_7cycles-c}
    \end{subfigure}   
    \hfill
    \begin{subfigure}[b]{0.48\linewidth}
        \centering
        \includegraphics[width=\linewidth]{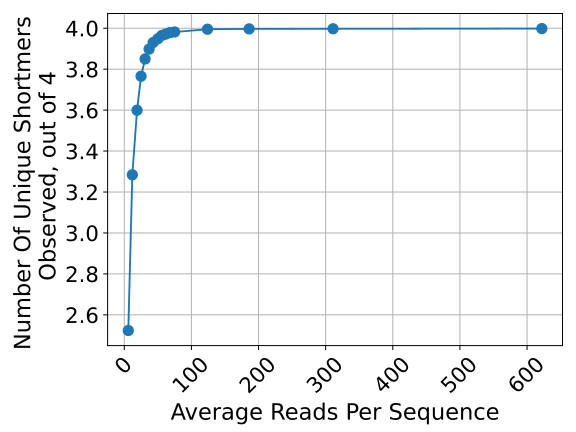}
        \caption{}
        \label{fig:large_scale_experiment_640BC_7cycles-d}
    \end{subfigure}

    \begin{subfigure}[b]{0.43\linewidth}
        \centering
        \includegraphics[width=\linewidth]{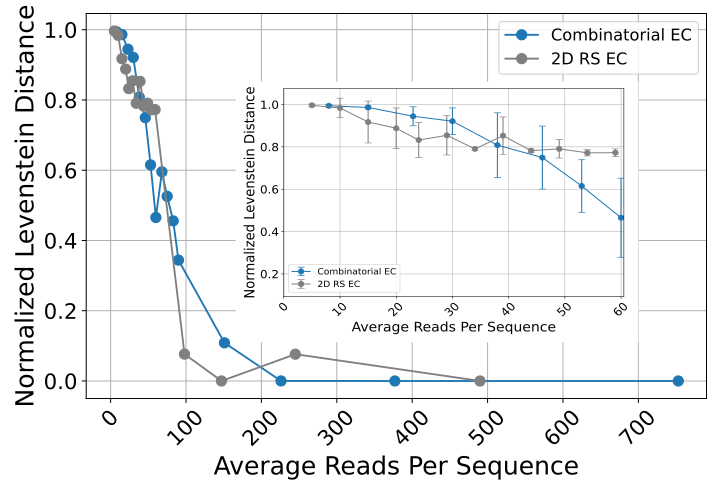}
        \caption{}
        \label{fig:large_scale_experiment_640BC_7cycles-e}
    \end{subfigure}   
    \hfill
    \begin{subfigure}[b]{0.48\linewidth}
        \centering
        \includegraphics[width=\linewidth]{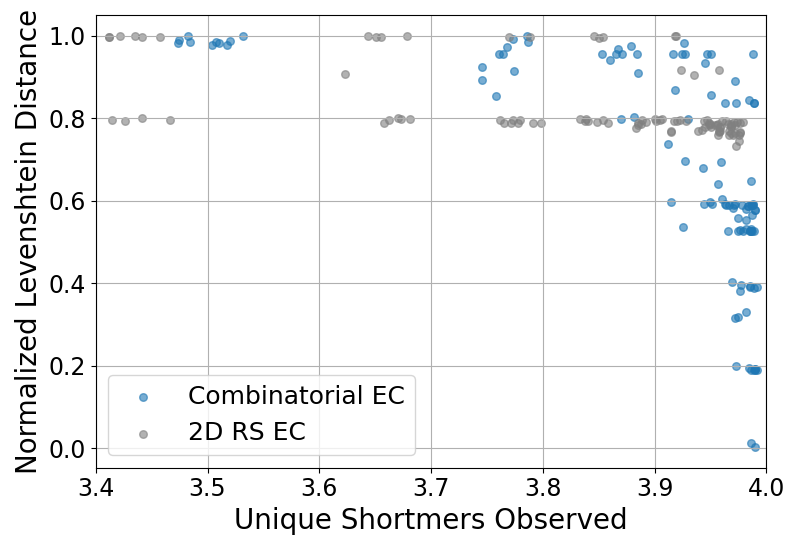}
        \caption{}
        \label{fig:large_scale_experiment_640BC_7cycles-f}
    \end{subfigure}

  \caption{Experiment analysis. (a) Read length distribution, (b) All universals start position in the sequences. (c) Read coverage distribution, with the x-axis representing the number of reads and the y-axis representing the number of barcodes. (d) X-axis is the average read per sequence, and the y-axis is the number of observed unique shortmers out of the 8, with $N=8$, $K=4$, $M=7$, 640 different sequences. (e) Comparison of decoding performance analysis, with the x-axis showing the average reads per sequence and the y-axis representing normalized byte-level Levenshtein distance. (f) Comparison of error correction performance, with the x-axis representing the number of unique shortmers observed and the y-axis showing normalized byte-level Levenshtein distance.}
  \label{fig:large_scale_experiment_640BC_7cycles 320bc}
\end{figure*}

Similar to the POC experiment, the analyzed sequences contained gaps of various lengths with unidentifiable sequences as illustrated in Fig.~\ref{fig:large_scale_experiment_640BC_7cycles-a} and Fig.~\ref{fig:large_scale_experiment_640BC_7cycles-b}. Therefore in this experiment we also used the same BLAST-based decoding algorithm described in Methods~\hyperref[Methods:640BC-8cycles]{BLAST-based analysis of large-scale combinatorial experiment}. The large-scale nature of the experiment resulted in coverage variability, with some sequences receiving high coverage while others remained sparse (See Fig.~\ref{fig:large_scale_experiment_640BC_7cycles-c}), directly impacting decoding accuracy.

To compare the two error-correction approaches, we split the 640 BCs in this experiment into two groups: BCs 1–320 were encoded using the traditional 2D RS-based method, whereas BCs 321–640 were encoded with our newly designed combinatorial error-correcting approach, serving as the inner code while sharing the same outer code as the former. In the combinatorial error-correcting scheme, we employed an RS(7,5) code to protect the VT syndromes. This enables the scheme to correct up to two missing shortmers (composite erasures) in each sequence. In the traditional 2D RS-based approach, we utilize an RS(7,6) code for inner error correction. Consequently, in both schemes, each sequence contains 36 bits of information and 6 bits of redundancy used for error-correction, meaning that both schemes use the same redundancy, allowing us to focus our analysis only on the error correction capabilities. Additional, both methods incorporated an RS(32,30) outer code to enhance overall reliability.  Thus, as each scheme had 300 sequences of information in total, the total encoded information was 10,640 bits per experiment.
Our analysis showed that down-sampling leads to an increase in the prevalence of erasure errors Fig.~\ref{fig:large_scale_experiment_640BC_7cycles-d}. Specifically in the critical region of low reads coverage where the average read count per sequence is below 50. This region also exhibit an increase decoding failure rate characterized by high Levenshtein distance (Fig.~\ref{fig:large_scale_experiment_640BC_7cycles-e}). Zooming in on the region with a high erasure rate reveals that, although both schemes employ the same redundancy, the Combi-VT code exhibits a slight advantage over the traditional 2D RS (Fig.~\ref{fig:large_scale_experiment_640BC_7cycles-e}, inset). When focusing the analysis to sequences where erasure errors where detected (regardless of sequencing coverage), we see that the suggested error correction clearly outperforms the traditional approach (Fig.~\ref{fig:large_scale_experiment_640BC_7cycles-f}).
Nonetheless, the presence of other error types, such as substitutions and insertions significantly impact the decoding rate, sometimes shadowing the effect of the different error correction schemes. It is thus clear that any error correction scheme for combinatorial DNA-based data storage systems must take an holistic approach that can effectively address all possible error types.

\begin{table}[ht]
\centering
\begin{tabular}{|l|l|}
\hline
\textbf{Parameters}& \textbf{Value} \\ 
\hline
Text size& 2.66KB\\ 
\hline
Number of total BC& 640\\ 
\hline
Number of BC for 2D RS& 320\\ 
\hline
Number of BC for new EC& 320\\ 
\hline
Number inner BC& 64\\ 
\hline
Number outer BC& 10\\ 
\hline
Sequence length/number of cycles (M)& 7\\ 
\hline
Redundancy bit & 6bit\\
\hline
Error correction method inner code & (7,6)\\ 
\hline
RS outer code& (32,30)\\ 
\hline
Oligo length (nuc)& 475 bp\\ 
\hline
BC length inner& 75 bp\\ 
\hline
BC lenth outer, Payload, Uni length& 25 bp\\ 
\hline
Number of unique k-mers (N)& 8\\ 
\hline
 Number of unique k-mers in each letter (K)&4\\\hline
 Logical density (bits/ Synthesis Cycle)&6\\\hline
 Alphabet size&$\binom{8}{4} = 70 \approx 64 = 2^6$\\\hline
\end{tabular}
\caption{Design parameters of the second large-scale experiment.}
\label{tab:Experiment parameters 320bc}
\end{table}

\begin{table}[ht]
\centering
\begin{tabular}{|l|l|}
\hline
\textbf{Condition} & \textbf{Value} \\ 
\hline
Number of reads & 3,615,438\\ \hline 
 Number of reads with 450bp+&1,590,536\\ \hline 
 Number of reads with BC identified&397,806\\

Expected BC& 640\\ 
\hline
Recoverd BC& 640\\ 
\hline
Recoverd BC before error correction& 632\\ 
\hline
BC not fully recovered& 8\\ \hline
\end{tabular}
\caption{Summary of the counts of sequencing reads analyzed in the second large-scale experiment.}
\label{tab:Summary of sequencing reads analyzed in the study 320bc}
\end{table}

\begin{table}[ht]
\centering
\begin{tabular}{|l|l|l|}
\hline
Outer BC& Inner BC &Average number of reads\\ 
\hline
1& 1-64&897\\ \hline 
 2&65-128&769\\ \hline 
 3&129-192&808\\ \hline 
4& 193-256&680\\ 
\hline
5& 257-320&610\\ 
\hline
6& 321-384&1103\\ 
\hline
7& 385-448&212\\ \hline
 8& 449-512&257\\\hline
 9& 513-576&292\\\hline
 10& 577-640&584\\\hline
\end{tabular}
\caption{Summary of the average number of reads per outer BC in the second large-scale experiment. Each outer BC was conducted in a different well.}
\label{tab:Summary of average length of outer BC 320bc}
\end{table}

\section*{Discussion}
In this study, we investigated the occurrence and correction of erasure errors in combinatorial DNA-based data storage, focusing on their prevalence during down-sampling. In the analysis of two previously published datasets, these errors were shown to become more pronounced as sequencing depth decreased, a characteristic feature of combinatorial DNA-based data storage. While these erasure errors can result from underlying base-level insertions, deletions, or substitutions that cause certain shortmers to become unidentifiable, our analysis indicates that they also manifest as higher-level phenomena reflecting the nature of the combinatorial synthesis. We further validated these observations using a large-scale experimental proof-of-concept, confirming that erasure errors are a dominant issue, particularly under low sequencing coverage. To address these errors, we developed an asymmetric error-correcting code specifically designed to mitigate their impact. Our method consistently outperformed traditional symmetric error-correcting codes, including 2D Reed-Solomon (RS), especially in scenarios characterized by high erasure error rates.

In comparison to previous studies, which primarily focused on symmetric error-correcting codes or sequence-level corrections, our approach directly targets the asymmetric nature of errors in combinatorial DNA. K-mer-level error correction therefore serves as a complementary mechanism to base-level correction, addressing asymmetric errors inherent to combinatorial DNA representations. This is an important advancement, as it acknowledges the unique error profile of combinatorial DNA, where the reconstruction of specific combinatorial letters is compromised due to the loss of included k-mers. By tailoring our method to this problem, we achieved significant improvements in decoding accuracy. It is especially interesting to compare the method proposed in this work to the one presented in previous large-scale combinatorial work. The most suitable for comparison is the combination of fountain code and forward code used in \cite{anavy2019data}. The two approach differ significantly, specifically in low coverage regimes. While our proposed code focuses on the asymmetric erasure errors that dominate low coverage scenarios, using fountain code in such cases results in complete data loss as was clearly demonstrated in Anavy \textit{et. al}. This further support the use of codes specifically tailored to the expected error patterns.

Our analysis revealed that down-sampling consistently led to an increase in erasure errors, particularly when the average coverage dropped below 50 reads per sequence. This critical threshold was identified in both experiments, where low sequencing depth directly impacted decoding accuracy. Additionally, the large-scale nature of our experiments introduced coverage variability, with some sequences receiving high read coverage while others remained sparse. This variability directly impacted decoding performance, emphasizing the need for error correction methods that can adapt to uneven coverage distributions. To accurately identify k-mers within each read, even in cases where sequencing errors or gaps disrupted the expected sequence structure, we employed a BLAST-based approach for sequence analysis. This choice was driven by the design of the combinatorial sequences, where each letter is represented by a set of predefined k-mers. The BLAST-based method ensured robust detection of these k-mers, even under challenging conditions, making it essential for maintaining decoding accuracy, especially in low-coverage scenarios.

We observed that our new error correction method demonstrated superior performance over the 2D RS approach, both in simulations and in the large-scale experiment, particularly in scenarios dominated by erasure errors. This performance difference can be attributed to the targeted design of our method, which leverages tensor-product (TP) codes to integrate standard erasure and substitution-correcting codes (such as Reed-Solomon (RS) codes) with Varshamov-Tenengolts (VT) codes. This design effectively addresses the asymmetric nature of errors in combinatorial DNA, where missing k-mers directly impact decoding success. As seen in the second experiment, the new error correction method proved particularly effective in recovering sequences from pools with low read counts, where traditional 2D RS failed to fully decode the data.

However, our study primarily focused on erasure errors, with other error types (e.g., substitutions, insertions, and deletions) not explicitly addressed. Our analysis showed that despite the focus on erasure errors, substitutions and insertions also affected decoding accuracy, highlighting the need for a unified error correction scheme that accounts for all possible error types. This is particularly important given the variability in sequencing quality and coverage observed in our experiments.

Overall, our findings emphasize the importance of developing unique error correction schemes for combinatorial DNA storage, particularly when the sequencing depth is limited. The adaptability of our asymmetric error-correcting code makes it suitable for a wide range of applications, beyond the datasets we analyzed. Moreover, the principles demonstrated here can be extended to other DNA-based data storage systems facing similar asymmetric error challenges.

In conclusion, our study provides a detailed understanding of erasure errors in combinatorial DNA storage and offers a practical solution for their correction. This advancement represents a significant step toward improving the reliability of DNA-based data storage systems. Our findings also raise important questions regarding the optimal integration of traditional error-correction schemes with specialized approaches tailored to combinatorial DNA. Further research is needed to develop unified error correction strategies that address the full spectrum of errors encountered in large-scale DNA storage systems. Future studies could focus on optimizing the current asymmetric error-correcting code to further enhance performance, while also exploring the impact of other error types (such as substitutions, insertions, and deletions) on decoding accuracy in combinatorial DNA. Moreover, developing alternative methods for k-mer detection beyond BLAST, capable of faster and more scalable analysis of large-scale combinatorial datasets, could significantly enhance decoding efficiency. Finally, optimizing the design of combinatorial sequences to improve k-mer detection and minimize sequencing errors could further boost system robustness.
From a broader perspective, as the technology matures, it will be important to establish more general evaluation metrics for DNA data storage systems. Such metrics would enable meaningful comparisons between combinatorial DNA synthesis and base-level synthesis approaches, which currently operate under fundamentally different data representations and error models. These metrics should account not only for error-correction performance, but also for system-level considerations such as writing time, cost, required sequencing coverage, and achievable information density.

\section*{Methods}

\subsection*{BLAST-based analysis of the data from Yan \textit{et al.}\cite{yan2023scaling}}
\phantomsection
\label{Methods:Data-Analysis-HelloWorld-1-cycles}

The 102,222 reads from the sequencing results were analyzed using the following procedure.
\begin{enumerate}
    \item \textbf{Read length filtering.} Reads where filtered to retain only those with length greater than a predefined threshold of 50 nt. This resulted in 22,477 reads.
    \item \textbf{BLAST alignment.} Each read was aligned against the design sequence using BLAST
    \begin{itemize}
        \item Generate a BLAST DB of the universal sequences
        
        \begin{lstlisting}
        > makeblastdb -in <universals.fasta> -dbtype nucl -out <blast_db>
        \end{lstlisting}
        \item Run blast with tolerance for small mismatches 
        and gaps to accommodate sequencing errors.        
        
        \begin{lstlisting}
        > blastn -query output_fasta_file -db <blast_db> -outfmt "6 qseqid sseqid pident length mismatch gapopen qstart qend sstart send evalue bitscore btop qseq sseq" -word_size 7 -gapopen 5 -gapextend 2 -reward 1 -penalty -3 -min_raw_gapped_score 18 -num_threads 6
        \end{lstlisting}

    \end{itemize}
    \item \textbf{Sequence structure extraction.} The table of the BLAST results was analyzed. For each read, the sequences directly left and right to the identified U1 sequence were compared to the barcode sequences and the payload sequences respectively. The comparison was based on Levenshtein distance and the nearest neighbor sequence was selected as long as the distance $d$ was $d\leq 6$.       
\end{enumerate}

A spreadsheet with the results of the BLAST-based analysis is available as Supplementary~\hyperref[sup:Analysis-output-of-Yan-data]{BLAST Analysis output of Yan \textit{et al.} dat}

\subsection*{BLAST-based analysis of large-scale combinatorial experiment}
\phantomsection
\label{Methods:640BC-8cycles}
The large-scale experimental data was analyzed using an extension of the analysis pipeline presented in \hyperref[Methods:Data-Analysis-HelloWorld-1-cycles]{BLAST-based analysis of the data from Yan \textit{et al.}\cite{yan2023scaling}} with the following adjustments.
A total of 3,855,223 reads where analyzed.

Although BLAST is generally considered a high computational cost algorithm, in our case the “reference genome” consisted only of the universal parts and the set of building-block shortmers. As a result, the effective genome length was very short (less than 600 bp), making the BLAST runtime negligible. Nevertheless, we acknowledge that further optimization could improve decoding performance, and exploring high-performance alignment and decoding methods remains an important direction for future work.

\begin{enumerate}
    \item \textbf{Read length filtering.} Reads where filtered to retain only those with length greater than a predefined threshold of 500 nt. This resulted in 2,136,524 reads. 
    
    \item \textbf{BLAST alignment and positional extraction.} BLAST was used with the same parameters while only the universal sequence database was altered to include all universal sequences in the large design. The output contains each read and their universals positions.


    \item \textbf{Extract universals positions in read.} For each read in the BLAST output, we return a list of universal-sequence positions by analyzing the alignment coordinates reported by BLAST for each universal-sequence. First, the best alignment per universal-sequence is selected based on the bit score. This list of positions is then passed through a correction function that ensures consistent spacing between consecutive universals, using the designed distances from the original barcode layout. That is, each universal-sequence must follow the previous sequence with a valid relative distance. The position for universal-sequences with no matches are assumed based on the previous sequence. This allows the analysis of the read even with some missing universal-sequences. To demonstrate the importance of this correction - while only 32,705 reads contained all 11 universal-sequences, after applying the correction all 8 payload-sequences where identified in as many as 154,071 reads.
    
\item \textbf{Barcode identification} Once the positions of the universal sequences are extracted, we estimate the barcode's location relative to the second universal-sequence (U2). Specifically, we anchor the read by subtracting the barcode length from the position of Universal 2 (U2), and then shift forward by the expected fixed offset: the length of a universal 1 and a 5' linker region See \ref{fig:large_scale_experiment_640BC_8cycles-a}. At this computed position, we attempt to match the barcode against a known inner barcode set, allowing for Levenshtein distance 6. If a valid inner barcode is found, we then check for a matching outer barcode in the prefix of the read with Levenstein distance 4. Only if both inner and outer barcodes are successfully identified, the barcode index is assigned to the read; otherwise, the read is excluded from downstream analysis.

\item \textbf{Payload identification} After the barcode is identified, we locate the payload sequences using the positions of the internal universal sequences. Specifically, we focus on the region between U3 and U11, which frame the payloads in our design. For each expected payload position, we compute the start location by shifting the universal position forward by the known universal length. At each position, we compare the sequence in the read to the expected payload set, using Levenshtein distance 6. If a match is found, the corresponding payload index is recorded. If no match is found or if the read is too short, the corresponding payload position is marked as undetected.


\item \textbf{Reconstruction and Decoding}. Once barcodes and payloads have been identified, we reconstruct the encoded message by aggregating and decoding reads according to their barcode group as follows:

\begin{itemize}
    \item \textbf{Barcode grouping}. Reads are grouped based on their identified barcode, such that each group corresponds to one unique combinatorial sequence.
    
    \item \textbf{Payload aggregation per position}. For each barcode group, and at each payload position, we look at all the reads and collect which payloads appear there. This gives a distribution of payloads for that position, which helps us understand which ones are most likely, even when there are sequencing errors. From the observed distribution, we select the $K$ most frequent payloads. The parameter $K$ is determined by the size of the combinatorial alphabet. If a tie occurs when selecting the top $K$ payloads, one of the candidates is selected arbitrarily. If fewer than $K$ payloads are identified, then we define the amount of payloads missing as asymmetric-erasure-errors. 
    
    \item \textbf{Combinatorial letter assignment}. The selected $K$ payloads are interpreted as a single combinatorial letter from the predefined alphabet. If no match exists, the letter is marked as an erasure.
    
    \item \textbf{Error-correcting}. The full reconstructed sequence of combinatorial letters is then decoded using a two-dimensional Reed–Solomon (2D RS) error-correcting code.
\end{itemize}
\end{enumerate}

\subsection*{Large-scale experiment comparing error-correction approaches}
\phantomsection
\label{Methods:experiment-2D-vs-newEC}
The large-scale experimental data was analyzed using the same analysis pipeline presented in REF with the following minor differences:
\begin{itemize}
    \item Each sequence contained only 7 payload-sequences and therefore the analysis focused on the region between U3-U10.
    \item The decoding of the error correction code was applied according to the experimental setting. The first 320 sequences were decoded using 2D RS while sequences 321-640 were decoded using the new combinatorial error-correction code.
\end{itemize}

\subsection*{Simulation of a combinatorial DNA-based storage system}
\label{Methods:Simulation-of-combinatorial-DNA-based-storage-system}

\subsubsection*{Synthesis and Sequencing Simulation with Errors}
\phantomsection
\label{Methods:Synthesis-and-Sequencing-Simulation-with-Errors}

\begin{itemize}
   
    \item \textbf{Simulation of the synthesis process:} DNA molecules corresponding to the designed sequences are generated using combinatorial \(k\)-mer DNA synthesis (see Fig.~\ref{fig:view of a combinatorial alphabet}).
    \begin{itemize}
        \item For each combinatorial sequence, determine the number of synthesized copies by sampling from $X \sim N(\mu=1000,\sigma^2=100)$. Let \(x\) be the number of copies for a specific sequence.
        \item For every position in the sequence, uniformly sample \(x\) independent \(k\)-mers from the set of member \(k\)-mers of the combinatorial letter in the specific position.
        \item Concatenate the sampled \(k\)-mers to the corresponding position of each of the \(x\) partial molecules, gradually building the full synthesized sequences.
    \end{itemize}

    \item \textbf{Error simulation:} We model synthesis and sequencing errors as follows:
        \begin{itemize}
            \item The probabilities of deletion, insertion, and substitution are specified by the parameters \(P_d\), \(P_i\), and \(P_s\), respectively.
            
            \item Deletion and insertion errors are considered synthesis-related and are applied at the level of entire \(k\)-mers — meaning a full \(k\)-mer may be deleted or inserted at a given position during the simulation.
            
            \item Substitution errors are modeled as sequencing errors and occur at the single-base level. A single base is replaced by another, independently of its position within a \(k\)-mer.
        \end{itemize}

    \item \textbf{Mixing:}. After synthesis, the molecules are randomly mixed to reflect natural molecular combinations. This is implemented by randomized data line shuffle using an SQLite-based method, which allows efficient shuffling even for large-scale datasets~\cite{sqlite}.

    \item \textbf{Reading and sampling:} To simulate the sequencing process, a subset of the synthesized molecules is sampled. The number of sampled reads is set to \(S \times (\text{number\_of\_synthesized\_sequences})\), representing the expected sequencing depth.

\end{itemize}

\subsubsection*{Reconstruction}
\phantomsection
\label{Methods:Simulation-Reconstruction}
The reconstruction and barcode decoding from sequencing reads was carried out using the following steps:

\begin{enumerate}
    \item Barcode decoding. Decode the barcode sequence of each read using Reed–Solomon (RS) error-correcting code.

    \item Grouping by barcode. Group the sequencing reads by their decoded barcode sequences, enabling the reconstruction of the combinatorial sequences.
    
    \item Filtering of read groups. Exclude barcode groups with fewer  than 10\% of the expected number of sampled reads.
    
    \item Combinatorial reconstruction. For each set of reads grouped by the same barcode:
    \begin{itemize}
        \item Traverse the sequence position by position.
        \item At each position, identify the $K$ most common \(k\)-mers across the reads.
        \item Use these frequent \(k\)-mers to determine the corresponding combinatorial letter \(\sigma\).
        \item Calculate the length deviation \(\Delta\) between each read and the expected sequence length.
        \item Exclude reads for which \(|\Delta| > k - 1\), as they likely contain synthesis or sequencing anomalies.
        \item Replace any \(k\)-mers that are not valid members of the combinatorial alphabet with an erasure symbol.
    \end{itemize}
    
    \item Handling missing barcodes. Replace missing barcodes with an erasure sequences to enable proper outer RS decoding.
    
    \item Normalized Byte-level Levenshtein distance calculation.
    \begin{itemize}
        \item Compute the Levenshtein distance between the observed sequence $O$ and the expected sequence $E$.
        \item Normalize this distance by dividing it by the length of the expected sequence:
        \[
        \text{Normalized Levenshtein Distance (O,E)} = \frac{\text{Levenshtein Distance (O,E)}}{\text{Length of Expected Sequence } |O|}
        \]
    \end{itemize}
\end{enumerate}

\subsection*{Code Construction}
\subsubsection*{Definitions and Notations}
For a positive integer $n$, let $[n]\triangleq\{0,\ldots,n-1\}$. For a binary vector $\mathbf{x}$, $w_H(\mathbf{x})$ denotes the Hamming weight (shortly the weight) of $\mathbf{x}$, which is the number of ones in $\mathbf{x}$. 

Let $\Sigma = \{A, C, G, T\}$ be the DNA alphabet and let $k \in \mathbb{N}$ be the $k$-mer (shortmer) length. We let $\mathcal{S} = \{ \mathbf{s}_0, \mathbf{s}_1, \ldots, \mathbf{s}_{N-1} \}$ be a set of $N>1$ \emph{shortmers}, $\mathbf{s}_i \in \Sigma^\ell$ for $i\in [N]$, which are indexed lexicographically. For $K<N$, we define the \emph{$K$-combinatorial composite alphabet of $\mathcal{S}$} as $\Sigma_K^{\mathcal{S}} \triangleq \{ \mathbf{x}_0^{\mathcal{S}}, \mathbf{x}_1^{\mathcal{S}}, \ldots, \mathbf{x}_{\binom{n}{K}-1}^{\mathcal{S}}\}$, where each \emph{combinatorial composite symbol} $\mathbf{x}_i^{\mathcal{S}}$, for $i\in [\binom{N}{K}]$ is a \emph{set} of $K$ different shortmers chosen from the shortmers set $\mathcal{S}$. For simplicity, the set of symbols in $\Sigma_K^{\mathcal{S}}$ can be abstracted as length-$n$ binary vectors of weight $w$ in which every bit indicates whether a shortmer in $\mathcal{S}$ belongs to the set. Thus, every $\mathbf{x}_i^{\mathcal{S}} \in \Sigma_K^N$ is mapped to its indicator binary vector, denoted by $\mathbf{x}_i\in\{0,1\}^n$ and note that $\sum_{j=0}^{n-1} x_{i, j} = K$. 
In the code construction below, we refer to the composite symbols in our alphabet by their binary vector representation and denote the set of length-$n$ binary vectors of weight $K$ by $\Sigma_K^N$.

A sequence of length $M$ over a composite alphabet $\Sigma_K^\mathcal{S}$ is denoted by $\mathcal{X}^{\mathcal{S}} = (\mathbf{x}_{i_0}^\mathcal{S},\ldots,\mathbf{x}_{i_{m-1}}^\mathcal{S}) \in (\Sigma_K^\mathcal{S})^M$. This sequence can be abstracted as an $M \times N$ binary matrix $\mathcal{X}$, in which each row is matched with its corresponding composite symbol from $\Sigma_K^{\mathcal{S}}$. 

For our construction, we consider \emph{the combinatorial composite-DNA channel},
which receives an $M \times N$ matrix $\mathcal{X}$, and outputs a noisy version of $\mathcal{X}$, denoted by $\mathcal{Y}$. Similarly, we denote the rows of the matrix $\mathcal{Y}$ by $\mathbf{y}_h$, $h \in [M]$, such that $\mathbf{y}_h$ is a noisy version of $\mathbf{x}_{i_h}$

Lastly, since the exact shortmers in the set $\mathcal{S}$ do not matter but only the number of shortmers, we refer to the combinatorial composite alphabet from now on by the set $\Sigma_K^N$ and a length-$M$ sequence is simply an $M\times N$ matrix in $\Sigma_K^{M\times N}$, where  $\Sigma_K^{M\times N}$ refers to the set of all $M\times N$ matrices in which the weight of every row is $K$.

\begin{definition} \label{def:1}
\textbf{Composite asymmetric errors.} For a positive integer $e$ and a row vector $\mathbf{x}_i = (x_0, \ldots, x_{n-1}) \in \Sigma_K^N$, we say that the corresponding channel output $\mathbf{y}_i = (y_0, \ldots, y_{n-1}) \in \Sigma_{K-e}^N $, suffers from \emph{$e$ composite-asymmetric errors} if $y_i \le x_i$, and $\sum_{i=0}^{n-1} y_i = K-e.$
\end{definition}

Definition~\ref{def:1} can be extended to matrices (sequences) as described below.

\begin{definition}
\textbf{$(t,e)$-composite asymmetric errors.} For positive integers $e$ and $t$ and a matrix $\mathcal{X} =      \left( \mathbf{x}_0,  \ldots  ,\mathbf{x}_{M-1} \right)^T \in \Sigma_K^{M\times N}$, we say that the channel output matrix $\mathcal{Y} =       \left( \mathbf{y}_0, \ldots, \mathbf{y}_{M-1} \right)^T, $ suffers from \emph{$(t,e)$-composite asymmetric errors} if at most $t$ rows of $\mathcal{X}$ are noisy, each of them suffers from at most $e$ composite-asymmetric errors.
\end{definition}

A \emph{length-$m$ $(n,w)$-composite code $\mathcal{C}$} is a set of matrices over $\Sigma_w^{m\times n}$ and every codeword in $\mathcal{C}$ is referred {to} as a \emph{composite codeword}. We say that a length-$m$ $(n,w)$-composite code is a \emph{$(t,e)$-composite-asymmetric ECC (in short $(t,e)$-CAECC)}, if it can correct any $(t,e)$-composite asymmetric error. Such a code will be referred as an \emph{$[m,(n,w);t,e]$-composite code}.


We denote by $A(m,n,w)$ the size of the set of all binary matrices of dimension $m \times n$, in which each row is of weight exactly $w$, that is  $A(m, n, w) =  |\Sigma_w^{m \times n} |= \binom{n}{w}^m$. We denote by $A(m, n, w; t, e)$ the size of the largest $[m,(n,w);t,e]$-composite code. For a composite code $\mathcal{C} \subseteq \Sigma_w^{m\times n}$, we define its redundancy to be $r(\mathcal{C}) \triangleq \log(|A(m,n,w)|) - \log(|\mathcal{C}|)$. Furthermore, we denote by $r(m,n,w; t, e)$ the minimum redundancy of such a composite code.

\subsubsection*{Mathematical Definition of the Code}
To define our code construction, we first define some parameters of the VT code~\cite{varshamov1965code}. Given a length-$N$ binary vector $\mathbf{x} = (x_0, x_1, \ldots, x_{n-1}) \in \{0,1\}^N$, positive integers $\ell$ and $p$, we define the \emph{$\ell$-VT-syndrome over $p$}~\cite{varshamov1965code} of $\mathbf{x}$, denoted by $\mathbf{s}_{\ell}^p(\mathbf{x})$, as follows 
$\mathbf{s}_{\ell}^p(\mathbf{x}) \triangleq \sum_{i=0}^{N-1} i^\ell x_i \bmod p.$ Note that we have that, $\mathbf{s}_{\ell}^p(X) \in F_p$. For the simplicity of our construction, we assume that $p$ is the smallest prime such that $p\ge N$. According to Bertrand's Postulate~\cite{B1845}, it is known that that $N\leq p\leq 2N$.

Furthermore, for an integer $e \ge 1 $, we also define the \emph{$e$-complete-VT-syndrome over $p$}, denoted by $\textbf{S}_{e}^p(\mathbf{x})$, to be
$ \textbf{S}_{e}^p(\mathbf{x}) \triangleq (\textbf{s}_{1}^p(\mathbf{x}),\textbf{s}_{2}^p(\mathbf{x}),\ldots,\textbf{s}_{e}^p(\mathbf{x})),$ 
and note that $\mathbf{S}_{e}^p(\mathbf{x})\in F_{p^e}$. This definition can be extended to matrices by their rows as follows. Given a matrix $\mathbf{X} \in \Sigma_K^{M\times N}$, whose rows are given by $\mathbf{x}_0, \mathbf{x}_1, \ldots, \mathbf{x}_{m-1}$, and define the \emph{$e$-complete-VT-syndrome-vector over $p$} of $\mathbf{X}$, denoted by $\mathbf{S}_e^p(\mathbf{X})$, to be the vector in which its $i$-th entry corresponds to the $e$-complete-VT-syndrome of the $i$-th row in $\mathbf{X}$. That is, $\mathbf{S}_e^p(\mathbf{X})  \triangleq (\mathbf{S}_e^p(\mathbf{x}_0), \mathbf{S}_e^p(\mathbf{x}_1), \ldots, \mathbf{S}_e^p(\mathbf{x}_{m-1})) \in (F_{p^e})^{m}.$

Using the above definitions, the code construction is given below. 

\textbf{Code Construction}
Let $e\ge 1$, $p$ be the smallest prime number such that $p\ge N$, and $\mathcal{C}_t$ be an $[M, K, t+1]$ code over $F_{p^e}$ capable of correcting $t$ erasures. If $M\le p^e$, then $\mathcal{C}_t$ is selected as an MDS code with $K=M-t$. The code $\mathcal{C}_{M, N, K}^{(t,e)}$ is defined as follows. 
 $\mathcal{C}_{M, N, K}^{(t,e)} =  \{ \mathbf{X} \in \Sigma_K^{M\times N} : \mathbf{S}_e^p \left( \mathbf{X} \right)\in \mathcal{C}_t \}.$

To construct a composite-asymmetric error-correcting code, we first define the following parameters. Let \( e \geq 1 \) be a positive integer, and let \( p \) be the smallest prime number such that \( p \geq n \). We consider an \([m, k, t + 1]\) code \(\mathcal{C}_t\) over the finite field \(\mathbb{F}_{p^e}\), capable of correcting \(t\) erasures. When \( m \leq p^e \), the code \(\mathcal{C}_t\) is selected as a maximum distance separable (MDS) code with dimension \( k = m - t \).

Using these ingredients, we define the composite code as follows:

\begin{construction}
Let \( e \geq 1 \), \( p \) be the smallest prime number such that \( p \geq n \), and let \(\mathcal{C}_t\) be an \([m, m - t, t + 1]\) MDS code over \(\mathbb{F}_{p^e}\). The composite code is defined by:
\[
\mathcal{C}_{m,n,w}^{(t,e)} = \left\{ \mathbf{X} \in \Sigma_w^{m \times n} : \mathbf{S}_e^p(\mathbf{X}) \in \mathcal{C}_t \right\},
\]
where \(\mathbf{S}_e^p(\mathbf{X})\) denotes the \(e\)-complete-VT-syndrome-vector of the matrix \(\mathbf{X}\).
\end{construction}

\begin{theorem}
The code \(\mathcal{C}_{m,n,w}^{(t,e)}\) is a \((t, e)\)-composite-asymmetric error-correcting code (\((t,e)\)-CAECC).
\end{theorem}

\begin{proof}
Consider an \( m \times n \) matrix \(\mathbf{Y}\), obtained from a codeword \(\mathbf{X} \in \mathcal{C}_{m,n,w}^{(t,e)}\), where \(\mathbf{Y}\) suffers from \(t\) composite-asymmetric errors. Denote by \( 1 \leq i_1, \ldots, i_t \leq m \) the indices of the erroneous rows of \(\mathbf{Y}\), where each row has exactly \( w - e \) ones instead of \(w\). Without loss of generality, we prove the correction for the first erroneous row, \( \mathbf{y}_{i_1} \); the same argument applies to the others.

Since \(\mathcal{C}_t\) is capable of correcting \(t\) erasures, we can use its decoder to recover the correct \(e\)-complete-VT-syndrome of the \(i_1\)-th row, denoted by \( \mathbf{S}_e^p(\mathbf{x}_{i_1}) \). Consequently, we can also recover the individual \(\ell\)-VT-syndromes, \( s_\ell^p(\mathbf{x}_{i_1}) \), for \( 1 \leq \ell \leq e \).

Let \( h_1 < h_2 < \cdots < h_e \) be the set of indices corresponding to the positions in which \(\mathbf{y}_{i_1}\) differs from \(\mathbf{x}_{i_1}\), i.e., the locations of the asymmetric errors. By construction, the following relation holds for every \( 1 \leq \ell \leq e \):
\[
h_1^\ell + h_2^\ell + \cdots + h_e^\ell \equiv s_\ell^p(\mathbf{x}_{i_1}) - s_\ell^p(\mathbf{y}_{i_1}) \pmod{p}.
\]
Therefore, for each erroneous row, we obtain \( e \) such equations. As established in Theorem 1 in~\cite{D10}, these equations can be expressed as the elementary symmetric polynomials of the error locations, leading to a polynomial whose roots are precisely the error positions \( h_1, h_2, \ldots, h_e \).

The coefficients of this polynomial can be determined using the Newton–Girard identities~\cite{Seroul2000}, and its roots are uniquely defined over \(\mathbb{F}_p\) by Lagrange interpolation~\cite{Nathanson2010}. Thus, the set of error positions can be efficiently determined, and the erroneous row can be corrected by simply flipping the bits at these positions. Applying the same process to all \(t\) erroneous rows completes the decoding, which proves the theorem.
\end{proof}

\subsection*{Molecular biology experimental protocols}
\phantomsection
\label{Methods:Synthesis}
\subsubsection*{Synthesis and Sequencing}

Synthesis and sequencing of datasets were both made by Helixworks Technologies, Ltd, MTU Campus, Bishopstown, Rubicon Centre, Cork, Ireland, T12 Y275.

\phantomsection
\label{Methods:Molecular-biology-experimental-protocols}
\subsubsection*{Protocol for large scale proof of concept experiment}
\label{Methods:Barcoding-and-preparation-of-640BC-8cycles-oligonucleotides}
Each plate of 64 oligonucleotides was barcoded using the Oxford Nanopore Technologies (ONT) Native Barcoding Kit 96 V14 (SQK-NBD114.96), and pooled into a single vial. The prepared vials were ready for the adapter-ligation and clean-up steps, as outlined in ONT's Ligation Sequencing Amplicon protocol (SQK-NBD114.96).

\bibliography{sample}

\section*{Acknowledgements} This project has received funding by the European Union (DiDAX, 101115134). Views and opinions expressed are, however, those of the author(s) only and do not necessarily reflect those of the European Union or the European Research Council Executive Agency. Neither the European Union nor the granting authority can be held responsible for them. The authors thank Zohar Yakhini's research group and Eitan Yaakobi's research group for useful discussions. Lastly, the author thanks Nimesh Pinnamaneni from HelixWorks technology for useful discussion, and for his help in conducting this project. 
Preliminary results related to the constructions presented in Section~\hyperref[sec:Error-Correcting-Codes-for-Combinatorial-DNA]{Error-Correcting Codes for Combinatorial DNA} were presented at the IEEE International Symposium on Information Theory (ISIT) 2024~\cite{sabary2024error}.


\section*{Author contributions statement}
All authors designed the study. I.P. and O.S. implemented coding, simulations, and  experimental protocols. I.P and L.A. performed the data analysis. O.S. and R.G. led the formal mathematical investigation. Z.Y, E.Y, and L.A. supervised the study. I.P., O.S. and L.A. led the manuscript writing with contribution from all authors.

\section*{Additional information}
\textbf{Competing interests:} The authors declare no competing interests.


\section*{Data availability}
The raw data is available in ENA (European Nucleotide Archive). The datasets generated and/or analyzed during the current study are available in the ENA (European Nucleotide Archive) repository.
The dataset associated with Section~\hyperref[sec:Experimental-Proof-of-Concept]{Experimental Proof of Concept} is available under Accession Number PRJEB89959, and the dataset supporting Section~\hyperref[Experiment of 320BC 2D RS vs 320BC new error correction]{Experimental Comparison of the Suggested Code to 2D RS} is available under Accession Number PRJEB89961.

\section*{Code availability} 
Simulation of the errors and retrieval pipelines.
    \begin{itemize}
        \item Combi VT code: \href{https://github.com/InbalPreuss/dna_storage_simulation_new_error_correction.git}{GitHub link}
        \item 2D RS code: \href{https://github.com/InbalPreuss/dna_storage_simulation_2d_rs_EC_erasure.git}{GitHub link}
    \end{itemize}
Experiment design, including encoder and decoder.
\begin{itemize}
    \item Experiment of the first large-scale datasets (encoded with 2D RS): \href{https://github.com/InbalPreuss/dna_storage_experiment_helixworks_640bc.git}{GitHub link}
    \item Experiment of the first large-scale datasets (encoded with both 2D RS, and with combi VT code):  \href{https://github.com/InbalPreuss/dna_storage_experiment_helixworks_320bc_320bc.git}{GitHub link}
\end{itemize}
Data Analysis of previously published datasets.
\begin{itemize}
    \item Yan \textit{et al.}~\cite{yan2023scaling} experiment: \href{https://github.com/InbalPreuss/dna_storage_experiment_helloworld.git}{GitHub link}
    \item Preuss \textit{et al.}~\cite{preuss2024efficient}: \href{https://github.com/InbalPreuss/dna_storage_experiment.git}{GitHub link}
\end{itemize}

\clearpage
\setcounter{figure}{0}
\setcounter{table}{0}
\renewcommand{\figurename}{Supplementary Figure}
\renewcommand{\tablename}{Supplementary Table}

\section*{Supplementary Information}
\subsection*{\texorpdfstring{Summary of sequence counts analyzed in the studies conducted by Yan \textit{et al.}~\cite{yan2023scaling} and Preuss \textit{et al.}~\cite{preuss2024efficient}}{Summary of sequence counts from Yan et al. and Preuss et al.}}
\phantomsection
\label{Summary-of-sequnce-counts-analyzed-in-the-studies-conducted-by-Yan--Preuss}
In Table~\ref{sup:tab:Summary of barcode counts analyzed in the study}, we present a comparative summary of sequence counts from the studies conducted by Yan \textit{et al.}\cite{yan2023scaling} and Preuss \textit{et al.}\cite{preuss2024efficient}. These counts highlight the differences in data scope and analytical approach between the two studies~\cite{yan2023scaling,preuss2024efficient}.

\begin{table}[ht]
\centering
\begin{tabular}{|l|l|l|}
\hline
\textbf{BC} & \textbf{Yan \textit{et al.}~\cite{yan2023scaling} Count} & \textbf{Preuss \textit{et al.}~\cite{preuss2024efficient} Count}\\ 
\hline
1 & 924 & 1,365,295\\ \hline 
2 & 535 &  768,755\\ \hline 
3 & 1,073 & \\ \hline
4 & 1,219 & \\ \hline
5 & 901 & \\ \hline
6 & 466 & \\ \hline
7 & 1,027 & \\ \hline
8 & 1,796 & \\ \hline
\end{tabular}
\caption{Summary of barcode counts analyzed in the studies Yan \textit{et al.}~\cite{yan2023scaling} and Preuss \textit{et al.}~\cite{preuss2024efficient}.}
\label{sup:tab:Summary of barcode counts analyzed in the study}
\end{table}

\subsection*{BLAST Analysis output of Yan \textit{et al.} data}
\phantomsection
\label{sup:Analysis-output-of-Yan-data}
The output of the BLAST analysis that uses the Method~\hyperref[Methods:640BC-8cycles]{BLAST-based analysis of large-scale combinatorial experiment} of the first 5 steps, before the Reconstruction and Decoding, is provided as the ancillary file \path{results_good_reads_with_len_bigger_then_y_helloworld.zip}.

\subsection*{Experimental proof of concept 9BC that did not recover out of the 640BC}
\phantomsection
\label{sup:Experimental-proof-of-concept-9BC-that-did-not-recover-out-of-the-640BC}

As detailed in Section~\hyperref[sec:Experimental-Proof-of-Concept]{Experimental Proof of Concept}, when analyzing 100\% of the data, we examined the nine sequences that were not fully recovered, as shown in Fig.\ref{fig:9bc-not-recoverd}. In Fig.\ref{fig:9bc-not-recoverd-a}, only three out of the four payloads were recovered, with the count for an incorrect payload exceeding that of the correct payload. In Fig.\ref{fig:9bc-not-recoverd-b}, only three payload sequences were observed, resulting in an erasure error. Finally, in Fig.\ref{fig:9bc-not-recoverd-c}, more than four payloads were observed in specific cycles, with a tie for the fourth most frequent payload. This necessitated selecting which payload sequence to retain, occasionally leading to the incorrect selection of four payloads and, consequently, the wrong letter, $\sigma$.

\begin{figure*}[htbp]
  \centering
  \begin{subfigure}[b]{0.48\linewidth}
    \centering
    \includegraphics[width=\linewidth]{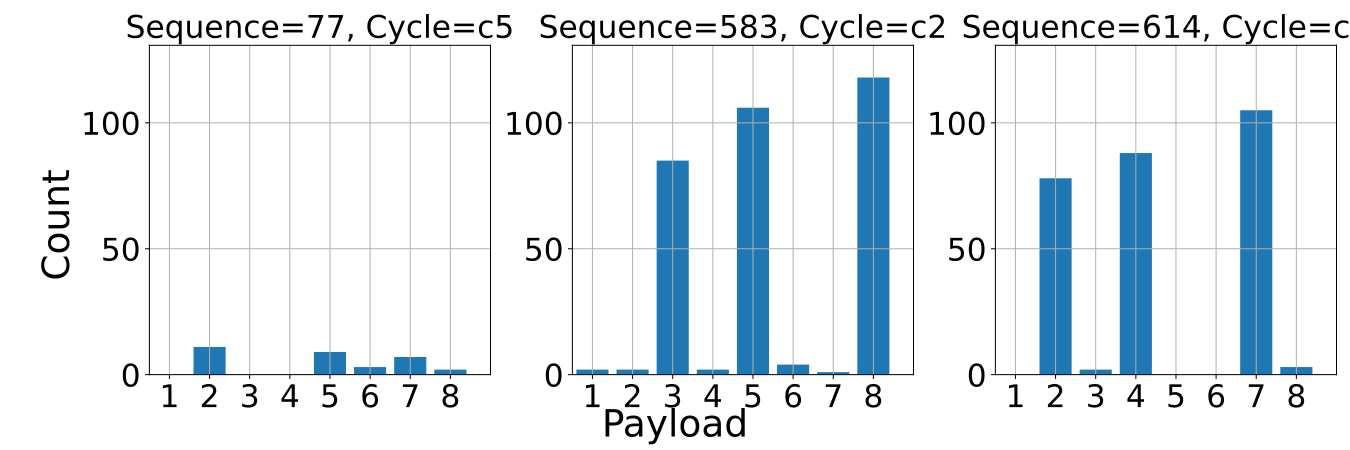}
    \caption{}
    \label{fig:9bc-not-recoverd-a}
  \end{subfigure}
  \hfill 
  \begin{subfigure}[b]{0.48\linewidth}
    \centering
    \includegraphics[width=\linewidth]{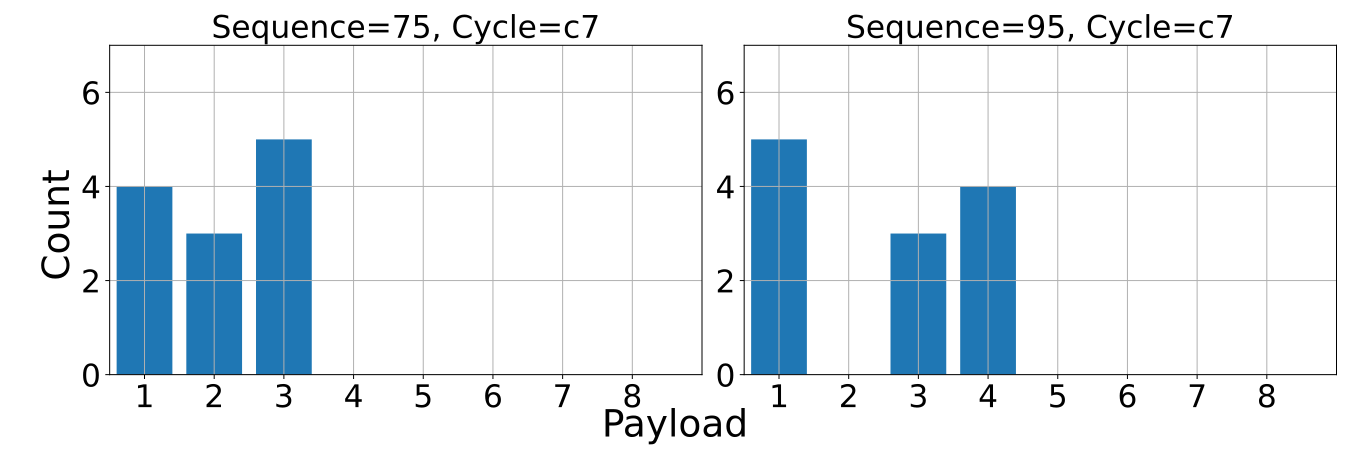}
    \caption{}
    \label{fig:9bc-not-recoverd-b}
  \end{subfigure}

  \begin{subfigure}[b]{0.48\linewidth}
    \centering
    \includegraphics[width=\linewidth]{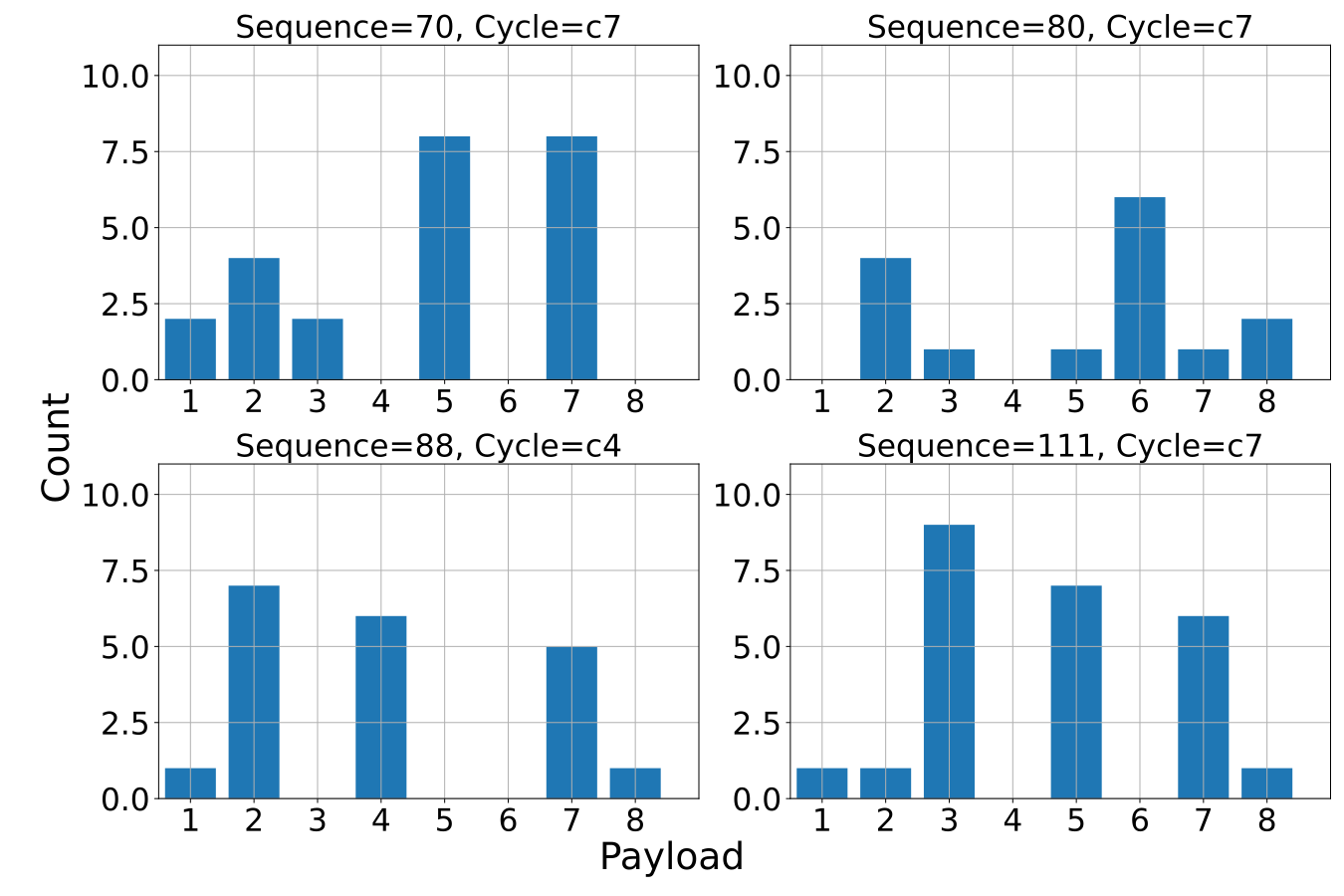}
    \caption{}
    \label{fig:9bc-not-recoverd-c}
  \end{subfigure}
  \caption{Analysis of the nine sequences that were not fully recovered. (a) Histograms showing only three out of the four payloads observed, with the incorrect payload having a higher count than the correct payload. (b) Histograms displaying three payloads recovered instead of four, resulting in an erasure error due to the missing payload. (c) Histograms illustrating more than four payloads in specific cycles, with a tie for the fourth most frequent payload. The incorrect selection of the four payloads led to an error in determining the corresponding letter, $\sigma$.}
  \label{fig:9bc-not-recoverd}
\end{figure*}

\end{document}